\documentclass[a4paper,UKenglish,cleveref, autoref]{lipics-v2019}

\title{Metric Dimension Parameterized by Treewidth}
\titlerunning{Metric Dimension Parameterized by Treewidth}

\author{\'{E}douard Bonnet}{Univ Lyon, CNRS, ENS de Lyon, Université Claude Bernard Lyon 1, LIP UMR5668, France}{edouard.bonnet@ens-lyon.fr}{https://orcid.org/0000-0002-1653-5822}{}

\author{Nidhi Purohit}{Univ Lyon, CNRS, ENS de Lyon, Université Claude Bernard Lyon 1, LIP UMR5668, France}{nidhi.purohit@ens-lyon.fr}{https://orcid.org/0000-0003-4869-0031}{}

\authorrunning{\'E. Bonnet, N. Purohit}

\Copyright{Édouard Bonnet, Nidhi Purohit}



\category{}

\relatedversion{}

\supplement{}

\acknowledgements{}

\nolinenumbers 

\EventEditors{John Q. Open and Joan R. Access}
\EventNoEds{2}
\EventLongTitle{42nd Conference on Very Important Topics (CVIT 2016)}
\EventShortTitle{CVIT 2016}
\EventAcronym{CVIT}
\EventYear{2016}
\EventDate{December 24--27, 2016}
\EventLocation{Little Whinging, United Kingdom}
\EventLogo{}
\SeriesVolume{42}
\ArticleNo{23}

\usepackage{xspace}
\usepackage[table]{xcolor}
\usepackage{tikz}
\usepackage{array}
\usetikzlibrary{fit}
\usetikzlibrary{calc}
\usetikzlibrary{shapes}
\usetikzlibrary{decorations.pathreplacing}
\usepackage{algpseudocode,algorithm,algorithmicx}

\algrenewcommand\algorithmicrequire{\textbf{Precondition:}}
\algrenewcommand\algorithmicensure{\textbf{Postcondition:}}
\usepackage{relsize}
\usepackage{amsmath,amssymb,amsthm}
\usepackage{caption}
\usepackage{subcaption}
\usepackage{flexisym}
\usepackage{framed}
\usepackage{float}
\usepackage{graphicx}
\usepackage{verbatim}

\makeatletter
\newcommand\footnoteref[1]{\protected@xdef\@thefnmark{\ref{#1}}\@footnotemark}
\makeatother

\newcommand{\kmIS}{\textsc{$k$-Multicolored Independent Set}\xspace}
\newcommand{\skmIS}{\textsc{$k$-MIS}\xspace}
\newcommand{\md}{\textsc{Metric Dimension}\xspace}
\newcommand{\gmd}{\textsc{$k$-Multicolored Resolving Set}\xspace}

\newcommand{\smd}{\textsc{MD}\xspace}
\newcommand{\sgmd}{\textsc{$k$-MRS}\xspace}

\newcommand{\dist}{\text{dist}}
\newcommand{\diam}{\text{diam}}
\newcommand{\tw}{\text{tw}}
\newcommand{\pw}{\text{pw}}
\newcommand{\ctw}{\text{ctw}}
\newcommand{\tl}{\text{tl}}

\newcommand{\swle}{$\text{sw}^{j}_{i}$\xspace}
\newcommand{\sele}{$\text{se}^{j}_{i}$\xspace}
\newcommand{\nwle}{$\text{nw}^{j}_{i}$\xspace}
\newcommand{\nele}{$\text{ne}^{j}_{i}$\xspace}

\newcommand{\swge}{$\text{sw}^{j}_{i}$\xspace}
\newcommand{\sege}{$\text{se}^{j}_{i}$\xspace}
\newcommand{\nwge}{$\text{nw}^{j}_{i}$\xspace}
\newcommand{\nege}{$\text{ne}^{j}_{i}$\xspace}

\newcommand{\sw}{\text{sw}^{j}_{i}\xspace}
\newcommand{\se}{\text{se}^{j}_{i}\xspace}
\newcommand{\nw}{\text{nw}^{j}_{i}\xspace}
\newcommand{\nee}{\text{ne}^{j}_{i}\xspace}

\newcommand{\tle}{\text{tl}^{j}_{i}\xspace}
\newcommand{\tri}{\text{tr}^{j}_{i}\xspace}
\newcommand{\ble}{\text{bl}^{j}_{i}\xspace}
\newcommand{\bri}{\text{br}^{j}_{i}\xspace}

\newcommand{\tc}{\text{tc}^{j}_{i}\xspace}
\newcommand{\bc}{\text{bc}^{j}_{i}\xspace}

\newcommand{\tbr}{\text{tb}^{j}_{i}\xspace}
\newcommand{\bbr}{\text{bb}^{j}_{i}\xspace}

\theoremstyle{plain}

\newcommand{\defparproblem}[4]{
 \vspace{1mm}
\noindent\fbox{
 \begin{minipage}{0.96\textwidth}
 \begin{tabular*}{\textwidth}{@{\extracolsep{\fill}}lr} #1 & {\bf{Parameter:}} #3 \\ \end{tabular*}
 {\bf{Input:}} #2 \\
 {\bf{Question:}} #4
 \end{minipage}
 }
 \vspace{1mm}
}

\begin{document}

\maketitle

\begin{abstract}
  A resolving set $S$ of a graph $G$ is a subset of its vertices such that no two vertices of $G$ have the same distance vector to $S$.
  The \textsc{Metric Dimension} problem asks for a resolving set of minimum size, and in its decision form, a resolving set of size at most some specified integer.
  This problem is NP-complete, and remains so in very restricted classes of graphs.
  It is also W[2]-complete with respect to the size of the solution.
  \textsc{Metric Dimension} has proven elusive on graphs of bounded treewidth.
  On the algorithmic side, a polytime algorithm is known for trees, and even for outerplanar graphs, but the general case of treewidth at most two is open.
  On the complexity side, no parameterized hardness is known.
  This has led several papers on the topic to ask for the parameterized complexity of \textsc{Metric Dimension} with respect to treewidth.
  
  We provide a first answer to the question.
  We show that \textsc{Metric Dimension} parameterized by the treewidth of the input graph is W[1]-hard.
  More refinedly we prove that, unless the Exponential Time Hypothesis fails, there is no algorithm solving \textsc{Metric Dimension} in time $f(\text{pw})n^{o(\text{pw})}$ on $n$-vertex graphs of constant degree, with $\text{pw}$ the pathwidth of the input graph, and $f$ any computable function.
  This is in stark contrast with an FPT algorithm of Belmonte et al. [SIAM J. Discrete Math. '17] with respect to the combined parameter $\text{tl}+\Delta$, where $\text{tl}$ is the tree-length and $\Delta$ the maximum-degree of the input graph. 
\end{abstract}

\section{Introduction}

The \md problem has been introduced in the 1970s independently by Slater \cite{Slater75} and by Harary and Melter \cite{Harary76}.
Given a graph $G$ and an integer $k$, \md asks for a subset $S$ of vertices of $G$ of size at most $k$ such that every vertex of $G$ is uniquely determined by its distances to the vertices of $S$.
Such a set $S$ is called a \emph{resolving set}, and a resolving set of minimum-cardinality is called a \emph{metric basis}.
The metric dimension of graphs finds application in various areas including network verification \cite{Beerliova06}, chemistry \cite{Chartrand00}, robot navigation \cite{Khuller96},  and solving the Mastermind game \cite{Chvatal83}.


\md is an entry of the celebrated book on intractability by Garey and Johnson~\cite{GareyJ79} where the authors show that it is NP-complete.
In fact \md remains NP-complete in many restricted classes of graphs such as planar graphs~\cite{Diaz17}, split, bipartite, co-bipartite graphs, and line graphs of bipartite graphs~\cite{Epstein15}, graphs that are both interval graphs of diameter two and permutation graphs~\cite{Foucaud17}, and in a subclass of unit disk graphs~\cite{Hoffmann12}.
On the positive side, the problem is polynomial-time solvable on trees~\cite{Slater75, Harary76, Khuller96}.
Diaz et al.~\cite{Diaz17} generalize this result to outerplanar graphs.
Fernau et al.~\cite{Fernau15} give a polynomial-time algorithm on chain graphs.
Epstein et al.~\cite{Epstein15} show that \md (and even its vertex-weighted variant) can be solved in polynomial time on co-graphs and forests augmented by a constant number of edges.
Hoffmann et al.~\cite{Hoffmann16} obtain a linear algorithm on cactus block graphs.

Hartung and Nichterlein~\cite{Hartung13} prove that \md is W[2]-complete (parameterized by the size of the solution $k$) even on subcubic graphs.
Therefore an FPT algorithm solving the problem is unlikely.
However Foucaud et al.~\cite{Foucaud17} give an FPT algorithm with respect to $k$ on interval graphs.
This result is later generalized by Belmonte et al.~\cite{Belmonte17} who obtain an FPT algorithm with respect to $\tl+\Delta$ (where $\tl$ is the tree-length and $\Delta$ is the maximum-degree of the input graph), implying one for parameter $\tl+k$.
Indeed interval graphs, and even chordal graphs, have constant tree-length.
Hartung and Nichterlein~\cite{Hartung13} presents an FPT algorithm parameterized by the vertex cover number, Eppstein~\cite{Eppstein15}, by the max leaf number, and Belmonte et al.~\cite{Belmonte17}, by the modular-width (a larger parameter than clique-width).

The complexity of \md parameterized by treewidth is quite elusive.
It is discussed~\cite{Eppstein15} or raised as an open problem in several papers~\cite{Belmonte17,Diaz17}.
On the one hand, it was not known, prior to our paper, if this problem is W[1]-hard.
On the other hand, the complexity of \md in graphs of treewidth at most two is still an open question.

\subsection{Our contribution}
We settle the parameterized complexity of \md with respect to treewidth.
We show that this problem is W[1]-hard, and we rule out, under the Exponential Time Hypothesis (ETH), an algorithm running in $f(\tw)|V(G)|^{o(\tw)}$, where $G$ is the input graph, $\tw$ its treewidth, and $f$ any computable function.
Our reduction even shows that an algorithm in time $f(\pw)|V(G)|^{o(\pw)}$ is unlikely on constant-degree graphs, for the larger parameter pathwidth $\pw$.
This is in stark contrast with the FPT algorithm of Belmonte et al. \cite{Belmonte17} for the parameter $\tl+\Delta$ where $\tl$ is the tree-length and $\Delta$ is the maximum-degree of the graph.
We observe that this readily gives an FPT algorithm for $\ctw+\Delta$ where $\ctw$ is the connected treewidth, since $\ctw \geqslant \tl$.
This unravels an interesting behavior of \md, at least on bounded-degree graphs: usual tree-decompositions are not enough for efficient solving.
Instead one needs tree-decompositions with an additional guarantee that the vertices of a same bag are at a bounded distance from each other.  

As our construction is quite technical, we chose to introduce an intermediate problem dubbed \gmd in the reduction from \kmIS to \md.
The first half of the reduction, from \kmIS to \gmd, follows a generic and standard recipe to design parameterized hardness with respect to treewidth. 
The main difficulty is to design an effective \emph{propagation gadget} with a constant-size left-right cut.
The second half brings some new local attachments to the produced graph, to bridge the gap between \gmd and \md.
Along the way, we introduce a number of gadgets: edge, propagation, forced set, forced vertex. 
They are quite streamlined and effective.
Therefore, we believe these building blocks may help in designing new reductions for \md.

\subsection{Organization of the paper}
In \cref{sec:prelim} we introduce the definitions, notations, and terminology used throughout the paper.
In \cref{sec:HighLevel} we present the high-level ideas to establish our result.
We define the \gmd problem which serves as an intermediate step for our reduction.
In \cref{sec:gmd} we design a parameterized reduction from the W[1]-complete \kmIS to \gmd parameterized by treewidth.
In \cref{sec:md} we show how to transform the produced instances of \gmd to \md-instances (while maintaining bounded treewidth).
In \cref{sec:perspectives} we conclude with some open questions.

\section{Preliminaries}\label{sec:prelim}
We denote by $[i,j]$ the set of integers $\{i,i+1,\ldots, j-1, j\}$, and by $[i]$ the set of integers $[1,i]$.
If $\mathcal X$ is a set of sets, we denote by $\cup \mathcal X$ the union of them.

\subsection{Graph notations}
All our graphs are undirected and simple (no multiple edge nor self-loop).
We denote by $V(G)$, respectively $E(G)$, the set of vertices, respectively of edges, of the graph $G$. 
For $S \subseteq V(G)$, we denote the \emph{open neighborhood} (or simply \emph{neighborhood}) of $S$ by $N_G(S)$, i.e., the set of neighbors of $S$ deprived of $S$, and the \emph{closed neighborhood} of $S$ by $N_G[S]$, i.e., the set $N_G(S) \cup S$.
For singletons, we simplify $N_G(\{v\})$ into $N_G(v)$, and $N_G[\{v\}]$ into $N_G[v]$.
We denote by $G[S]$ the subgraph of $G$ induced by $S$, and $G - S := G[V(G) \setminus S]$.
For $S \subseteq V(G)$ we denote by $\overline S$ the complement $V(G) \setminus S$.
For $A, B \subseteq V(G)$, $E(A,B)$ denotes the set of edges in $E(G)$ with one endpoint in $A$ and the other one in $B$.

The length of a path in an unweighted graph is simply the number of edges of the path.
For two vertices $u, v \in V(G)$, we denote by $\dist_G(u,v)$, the distance between $u$ and $v$ in $G$, that is the length of the shortest path between $u$ and $v$.
The diameter of a graph is the longest distance between a pair of its vertices.
The diameter of a subset $S \subseteq V(G)$, denoted by $\diam_G(S)$, is the longest distance between a pair of vertices in $S$.
Note that the distance is taken in $G$, \emph{not} in $G[S]$.
In particular, when $G$ is connected, $\diam_G(S)$ is finite for every $S$.
A \emph{pendant} vertex is a vertex with degree one.
A vertex $u$ is \emph{pendant to $v$} if $v$ is the only neighbor of $u$.
Two distinct vertices $u, v$ such that $N(u) = N(v)$ are called \emph{true twins}, and \emph{false twins} if $N[u] = N[v]$.
In particular, false twins are adjacent.
In all the above notations with a subscript, we omit it whenever the graph is implicit from the context.

\subsection{Treewidth, pathwidth, connected treewidth, and tree-length}
A \emph{tree-decomposition} of a graph $G$, is a tree $T$ whose nodes are labeled by subsets of $V(G)$, called \emph{bags}, such that for each vertex $v \in V(G)$, the bags containing $v$ induce a non-empty subtree of $T$, and for each edge $e \in E(G)$, there is at least one bag containing both endpoints of $e$.
A \emph{connected tree-decomposition} further requires that each bag induces a connected subgraph in $G$.
The width of a (connected) tree-decomposition is the size of its largest bag minus one.
The treewidth (resp. connected treewidth) of a graph $G$ is the minimum width of a tree-decomposition (resp. a connected tree-decomposition) of $G$.
The length of a tree-decomposition is the maximum diameter of its bags in $G$.
The tree-length of a graph $G$ is the minimum length of a tree-decomposition of $G$.
We denote the treewidth, connected treewidth, and tree-length of a graph by $\tw$, $\ctw$, and $\tl$ respectively. 
Since a connected graph on $n$ vertices has diameter at most $n-1$, it holds that $\ctw \geqslant \tl$.

The pathwidth is the same as treewidth except the tree $T$ is now required to be a path, and hence is called a path-decomposition.
In particular pathwidth is always larger than treewidth.
Later we will need to upper bound the pathwidth of our constructed graph.
Since writing down a path-decomposition is a bit cumbersome, we will rely on the following characterization of pathwidth.
Kirousis and Papadimitriou \cite{Kirousis85} show the equality between the interval thickness number, which is known to be pathwidth plus one, and the \emph{node searching number}.
Thus we will only need to show that the number of searchers required to win the following one-player game is bounded by a suitable function.
We imagine the edges of a graph to be contaminated by a gas.
The task is to move around a team of searchers, placed at the vertices, in order to clean all the edges.
A move consists of removing a searcher from the graph, adding a searcher at an unoccupied vertex, or displacing a searcher from a vertex to any other vertex (not necessarily adjacent).
An edge is cleaned when both its endpoints are occupied by a searcher.
However after each move, all the cleaned edges admitting a free-of-searchers path from one of its endpoints to the endpoint of a contaminated edge are recontaminated.
The node searching number is the minimum number of searchers required to win the game.

\subsection{Parameterized problems and algorithms}
Parameterized complexity aims to solve hard problems in time $f(k)|\mathcal I|^{O(1)}$, where $k$ is a parameter of the instance $\mathcal I$ which is hopefully (much) smaller than the total size of $\mathcal I$.
More formally, a \emph{parameterized problem} is a pair $(\Pi,\kappa)$ where $\Pi \subseteq L$ for some language $L \subseteq \Sigma^*$ over a finite alphabet $\Sigma$ (e.g., the set of words, graphs, etc.), and $\kappa$ is a mapping from $L$ to $\mathbb N$.
An element $\mathcal I \in L$ is called an \emph{instance} (or \emph{input}).
The mapping $\kappa$ associates each instance to an integer called \emph{parameter}.
An instance is said \emph{positive} if $\mathcal I \in \Pi$, and a \emph{negative} otherwise.
We denote by $|\mathcal I|$ the size of $\mathcal I$, that can be thought of as the length of the \emph{word} $\mathcal I$.
An \emph{FPT algorithm} is an algorithm which solves a parameterized problem $(\Pi,\kappa)$, i.e., decides whether or not an input $\mathcal I \in L$ is positive, in time $f(\kappa(\mathcal I))|\mathcal I|^{O(1)}$ for some computable function $f$.
We refer the interested reader to recent textbooks in parameterized algorithms and complexity~\cite{DowneyF13, Cygan15}.

\subsection{Exponential Time Hypothesis, FPT reductions, and W[1]-hardness}
The \emph{Exponential Time Hypothesis} (ETH) is a conjecture by Impagliazzo et al.~\cite{ImpagliazzoETH} asserting that there is no $2^{o(n)}$-time algorithm for \textsc{3-SAT} on instances with $n$ variables. 
Lokshtanov et al.~\cite{surveyETH} survey conditional lower bounds under the ETH.

An \emph{FPT reduction} from a parameterized problem $(\Pi \subseteq L,\kappa)$ to a parameterized problem $(\Pi' \subseteq L',\kappa')$ is a mapping $\rho: L \mapsto L'$ such that for every $\mathcal I \in L$:
\begin{itemize}
\item(1) $\mathcal I \in \Pi \Leftrightarrow \rho(\mathcal I) \in \Pi'$,
\item(2) $|\rho(\mathcal I)| \leqslant f(\kappa(\mathcal I))|\mathcal I|^{O(1)}$ for some computable function $f$, and
\item(3) $\kappa(\rho(\mathcal I)) \leqslant g(\kappa(\mathcal I))$ for some computable function $g$.
\end{itemize}
We further require that for every $\mathcal I$, we can compute $\rho(\mathcal I)$ in FPT time $h(\kappa(\mathcal I))|\mathcal I|^{O(1)}$ for some computable function $h$.
Condition (1) makes $\rho$ a valid reduction, condition (2) together with the further requirement on the time to compute $\rho(\mathcal I)$ make the mapping $\rho$ \emph{FPT}, and condition (3) controls that the new parameter $\kappa(\rho(\mathcal I))$ is bounded by a function of the original parameter $\kappa(\mathcal I)$.
One can therefore observe that using $\rho$ in combination with an FPT algorithm solving $(\Pi',\kappa')$ yields an FPT procedure to solve the initial problem $(\Pi,\kappa)$.

A standard use of an FPT reduction is to derive conditional lower bounds: if a problem $(\Pi,\kappa)$ is thought not to admit an FPT algorithm, then an FPT reduction from $(\Pi,\kappa)$ to $(\Pi',\kappa')$ indicates that $(\Pi',\kappa')$ is also unlikely to admit an FPT algorithm.
We refer the reader to the textbooks \cite{DowneyF13,Cygan15} for a formal definition of W[1]-hardness.
For the purpose of this paper, we will just state that W[1]-hard are parameterized problems that are unlikely to be FPT, and that the following problem is W[1]-complete even when all the $V_i$ have the same number of elements, say $t$ (see for instance \cite{Pietrzak03}).

\defparproblem{\kmIS (\skmIS)}{An undirected graph $G$, an integer $k$, and $(V_1, \ldots,V_k)$ a partition of $V(G)$.}{$k$}{Is there a set $I \subseteq V(G)$ such that $|I \cap V_i|~=1$ for every $i \in [k]$, and $G[I]$ is edgeless?}

Every parameterized problem that \kmIS FPT-reduces to is W[1]-hard. 
Our paper is thus devoted to designing an FPT reduction from \kmIS to \md parameterized by $\tw$.
Let us observe that the ETH implies that one (equivalently, every) W[1]-hard problem is not in the class of problems solvable in FPT time (FPT$\neq$W[1]).
Thus if we admit that there is no subexponential algorithm solving \textsc{3-SAT}, then \kmIS is not solvable in time $f(k)|V(G)|^{O(1)}$. 
Actually under this stronger assumption, \kmIS is not solvable in time $f(k)|V(G)|^{o(k)}$. 
A concise proof of that fact can be found in the survey on the consequences of ETH \cite{surveyETH}.

\subsection{Metric dimension, resolved pairs, distinguished vertices}

A pair of vertices $\{u,v\} \subseteq V(G)$ is said to be \emph{resolved} by a set $S$ if there is a vertex $w \in S$ such that $\dist(w,u) \neq \dist(w,v)$.
A vertex $u$ is said to be \emph{distinguished} by a set $S$ if for any $w \in V(G) \setminus \{u\}$, there is a vertex $v \in S$ such that $\dist(v,u) \neq \dist(v,w)$.
A \emph{resolving set} of a graph $G$ is a set $S \subseteq V(G)$ such that every two distinct vertices $u, v \in V(G)$ are resolved by~$S$.
Equivalently, a resolving set is a set $S$ such that every vertex of $G$ is distinguished by~$S$.
Then \md asks for a resolving set of size at most some threshold $k$.
Note that a resolving set of minimum size is sometimes called a \emph{metric basis} for $G$.

\defparproblem{\md (\smd)}{An undirected graph $G$ and an integer $k$.}{$\tw(G)$}{Does $G$ admit a resolving set of size at most $k$?}

Here we anticipate on the fact that we will mainly consider \md parameterized by treewidth.
Henceforth we sometimes use the notation $\Pi/\tw$ to emphasize that $\Pi$ is not parameterized by the natural parameter (size of the resolving set) but by the treewidth of the input graph.

 
\section{Outline of the W[1]-hardness proof of \md/$\tw$}\label{sec:HighLevel}

We will show the following.
\begin{theorem}\label{thm:main}
Unless the ETH fails, there is no computable function $f$ such that \md can be solved in time $f(\pw)n^{o(\pw)}$ on constant-degree $n$-vertex graphs.
\end{theorem}
We first prove that the following variant of \md is W[1]-hard.

\defparproblem{\gmd (\sgmd)}{An undirected graph $G$, an integer $k$, a set $\mathcal X$ of $q$ disjoint subsets of $V(G)$: $X_1, \ldots, X_q$, and a set $\mathcal P$ of pairs of vertices of $G$: $\{x_1,y_1\}, \ldots, \{x_h,y_h\}$.}{$\tw(G)$}{Is there a set $S \subseteq V(G)$ of size $q$ such that
  \begin{itemize}
  \item(i) for every $i \in [q]$, $|S \cap X_i|~= 1$, and
  \item(ii) for every $p \in [h]$, there is an $s \in S$ satisfying $\dist_G(s,x_p) \neq \dist_G(s,y_p)$?
  \end{itemize}}

In words, in this variant the resolving set is made by picking exactly one vertex in each set of $\mathcal X$, and not all the pairs should be resolved but only the ones in a prescribed set $\mathcal P$.
We call \emph{critical pair} a pair of $\mathcal P$.
In the context of \gmd, we call \emph{legal set} a set which satisfies the former condition, and \emph{resolving set} a set which satisfies the latter.
Thus a solution for \gmd is a legal resolving set.

The reduction from \kmIS starts with a well-established trick to show parameterized hardness by treewidth.
We create $m$ ``empty copies'' of the \skmIS-instance $(G,k,(V_1,\ldots,V_k))$, where $m := |E(G)|$ and $t := |V_i|$.
We force exactly one vertex in each color class of each copy to be in the resolving set, using the set $\mathcal X$.
In each copy, we introduce an edge gadget for a single (distinct) edge of $G$.
Encoding an edge of \skmIS in the \sgmd-instance is fairly simple: we build a pair (of $\mathcal P$) which is resolved by every choice but the one \emph{selecting both its endpoints} in the resolving set.
We now need to force a \emph{consistent choice of the vertex chosen in $V_i$} over all the copies.
We thus design a propagation gadget.
A crucial property of the propagation gadget, for the pathwidth of the constructed graph to be bounded, is that it admits a cut of size $O(k)$ disconnecting one copy from the other.
Encoding a choice in $V_i$ in the distances to four special vertices, called \emph{gates}, we manage to build such a gadget with constant-size ``left-right'' separator per color class.
This works by introducing $t$ pairs (of $\mathcal P$) which are resolved by the south-west and north-east gates but not by the south-east and north-west ones.
Then we link the vertices of a copy of $V_i$ in a way that the higher their index, the more pairs they resolve in the propagation gadget to their left, and the fewer pairs they resolve in the propagation gadget to their right.

We then turn to the actual \md problem.
We design a gadget which simulates requirement (i) by forcing a vertex of a specific set $X$ in the resolving set.
This works by introducing two pairs that are only resolved by vertices of $X$.
We attach this new gadget, called \emph{forcing set} gadget, to all the $k$ color classes of the $m$ copies.
Finally we have to make sure that a candidate solution resolves all the pairs, and not only the ones prescribed by $\mathcal P$.
For that we attach two adjacent ``pendant'' vertices to strategically chosen vertices.
One of these two vertices have to be in the resolving set since they are false twins, hence not resolved by any other vertex.
Then everything is as if the unique common neighbor $v$ of the false twins was added to the resolving set.
Therefore we can perform this operation as long as $v$ does not resolve any of the pairs of $\mathcal P$.

To facilitate the task of the reader, henceforth we stick to the following conventions:
 \begin{itemize}
 \item Index $i \in [k]$ ranges over the $k$ \emph{rows} of the \textsc{$k$-MRS/MD}-instance or color classes of \skmIS.
 \item Index $j \in [m]$ ranges over the $m$ \emph{columns} of the \textsc{$k$-MRS/MD}-instance or edges of \skmIS.
 \item Index $\gamma \in [t]$, ranges over the $t$ vertices of a color class.
 \end{itemize}
We invite the reader to look up \cref{tbl:glossary} when in doubt about a notation/symbol relative to the construction.
 
\section{Parameterized hardness of \gmd/$\tw$}\label{sec:gmd}

In this section, we give an FPT reduction from the W[1]-complete \kmIS to \gmd parameterized by treewidth.
More precisely, given a \textsc{$k$-Multicolored Independent Set}-instance $(G,k,(V_1, \ldots ,V_k))$ we produce in polynomial-time an equivalent \gmd-instance $(G',k',\mathcal X, \mathcal P)$ where $G'$ has pathwidth (hence treewidth) $O(k)$. 
  
 \subsection{Construction}
 
 Let $(G,k,(V_1, \ldots ,V_k))$ be an instance of \textsc{$k$-Multicolored Independent Set} where $(V_1, \ldots, V_k)$ is a partition of $V(G)$ and $V_i := \{v_{i,\gamma}$ $|$ $1 \leqslant \gamma \leqslant t\}$.
 We arbitrarily number $e_1, \ldots, e_j, \ldots, e_m$ the $m$ edges of $G$.

\subsubsection{Overall picture}
We start with a high-level description of the \sgmd-instance $(G',k',\mathcal X,\mathcal P)$.
For each color class $V_i$, we introduce $m$ copies $V_i^1, \ldots, V_i^j, \ldots, V_i^m$ of a \emph{selector gadget} to $G'$.
Each set $V_i^j$ is added to $\mathcal X$, so a solution has to pick exactly one vertex within each selector gadget.
One can imagine the vertex-sets $V_i^1, \ldots, V_i^m$ to be aligned on the \emph{$i$-th row}, with $V_i^j$ occupying the \emph{$j$-th column} (see \cref{fig:overall-picture}). 
Each $V_i^j$ has $t$ vertices denoted by $v^j_{i,1}, v^j_{i,2}, \ldots, v^j_{i,t}$, where each $v^j_{i,\gamma}$ ``corresponds'' to $v_{i,\gamma} \in V_i$.
We make $v^j_{i,1}v^j_{i,2} \ldots v^j_{i,t}$ a path with $t-1$ edges.

For each edge $e_j \in E(G)$, we insert an \emph{edge gadget} $\mathcal{G}(e_j)$ containing a pair of vertices $\{c_j,c_j'\}$ that we add to $\mathcal P$.
Gadget $\mathcal{G}(e_j)$ is attached to $V_i^j$ and $V_{i'}^j$, where $e_j \in E(V_i,V_{i'})$.
The edge gadget is designed in a way that the only legal sets that do \emph{not} resolve $\{c_j,c_j'\}$ are the ones that precisely pick $v^j_{i,\gamma} \in V_i^j$ and $v^j_{i',\gamma'} \in V_{i'}^j$ such that $e_j = v_{i,\gamma}v_{i',\gamma'}$.
We add a \emph{propagation gadget} $P_i^{j,j+1}$ between two consecutive copies $V_i^j$ and $V_i^{j+1}$, where the indices in the superscript are taken modulo $m$.
The role of the propagation gadget is to ensure that the choices in each $V_i^j$ ($j \in [m]$) corresponds to the same vertex in $V_i$. 

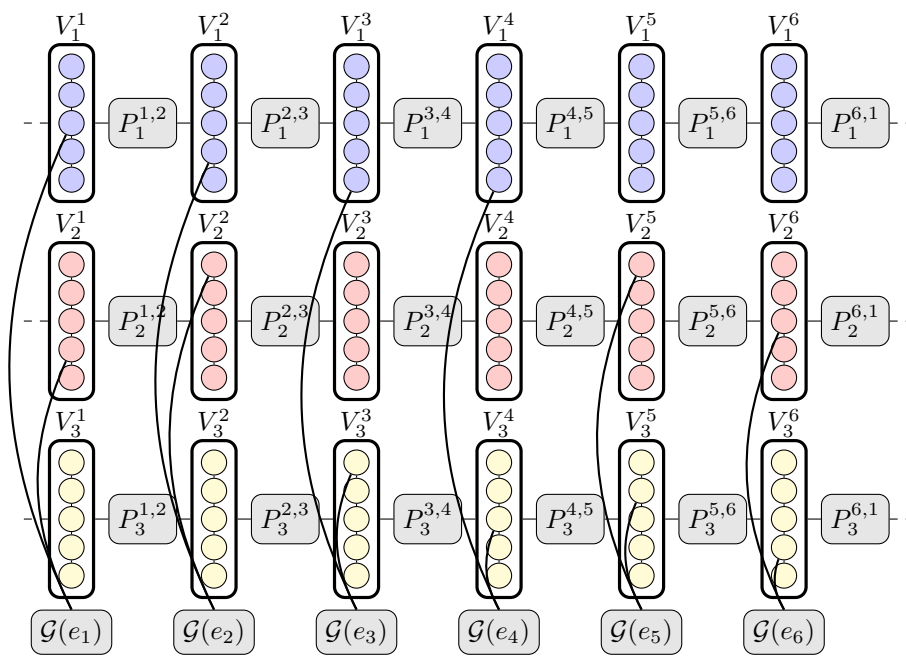
\begin{figure}[h!]
  \centering
  \begin{tikzpicture}[scale=0.75]
    \def\k{3}
    \def\t{5}
    \def\m{6}
    \def\s{0.5}
    \def\d{2.5}
    \def\x{0.5}

    \foreach \i/\c in {1/blue,2/red,3/yellow}{
      \foreach \j in {1,...,\m}{
        \foreach \h in {1,...,\t}{
          \node[draw,circle,fill=\c!20] (v\i\j\h) at (\d  * \j,\t * \k - \i * \k - \i * \x - \s * \h) {} ;
        }
        \node[draw, rectangle, minimum width=0.5cm, minimum height=1cm, rounded corners, very thick, fit=(v\i\j1) (v\i\j\t)] (v\i\j) {} ;
        \node (tv\i\j\t) at (\d * \j,\t * \k - \i * \k - \i * \x + \x - 0.3) {$V_\i^\j$} ; 
      }
    }
    \foreach \i in {1,...,\k}{
      \foreach \j in {1,...,\m}{
        \pgfmathtruncatemacro{\tm}{\t-1}
        \foreach \h in {1,...,\tm}{
          \pgfmathtruncatemacro{\hp}{\h+1}
          \draw (v\i\j\h) -- (v\i\j\hp) ;
        }
      }
    }
    \pgfmathsetmacro{\mm}{\m - 1}
    \foreach \i in {1,...,\k}{
      \foreach \j in {1,...,\mm}{
        \pgfmathtruncatemacro{\jp}{\j + 1}
        \node[draw, rectangle, rounded corners, fill=gray!20] (p\i\j) at (\d / 2 + \d * \j,\t * \k - \i * \k - \i * \x - \s * \t / 2 - \s / 2) {$P_\i^{\j,\jp}$} ;
        \draw (v\i\j) -- (p\i\j) -- (v\i\jp) ;
      }
    }
    \foreach \i in {1,...,\k}{
      \node[draw, rectangle, rounded corners, fill=gray!20] (p\i\m) at (\d / 2 + \d * \m,\t * \k - \i * \k - \i * \x - \s * \t / 2 - \s / 2) {$P_\i^{\m,1}$} ;
      \draw (p\i\m) -- (v\i\m) ;
      \draw[dashed] (p\i6.east)--++(0.5,0) ;
      \draw[dashed] (v\i1.west)--++(-0.5,0) ;
    }
    \foreach \j/\a/\b/\c/\e in {1/1/3/2/4,2/1/4/2/1,3/1/5/3/1,4/1/5/3/3,5/2/1/3/2,6/2/3/3/4}{
      \node[draw, rectangle, rounded corners, fill=gray!20] (e\j) at (\d * \j,1) {$\mathcal G(e_\j)$} ;
      \draw[thick] (e\j.north) to [bend left=25] (v\a\j\b) ;
      \draw[thick] (e\j.north) to [bend left=25] (v\c\j\e) ;
    }
  \end{tikzpicture}
  \caption{The overall picture with $k=3$ color classes, $t=5$ vertices per color class, $m=6$ edges, $e_1 = v_{1,3}v_{2,4}$, $e_2 = v_{1,4}v_{2,1}$, $e_3 = v_{1,5}v_{3,1}$, etc. The dashed lines on the left and right symbolize that the construction is cylindrical.}
  \label{fig:overall-picture}
\end{figure}

The intuitive idea of the reduction is the following.
We say that a vertex of $G'$ is \emph{selected} if it is put in the resolving set of $G'$, a tentative solution.
The propagation gadget $P_i^{j,j+1}$ ensures a consistent choice among the $m$~copies $V_i^1, \ldots, V_i^m$.  
The edge gadget ensures that the selected vertices of $G'$ correspond to an independent set in the original graph $G$.
If both the endpoints of an edge $e_j$ are selected, then the pair $\{c_j,c'_j\}$ is not resolved.
We now detail the construction. 
 
\subsubsection{Selector gadget}

For each $i \in [k]$ and $j \in [m]$, we add to $G'$ a path on $t-1$ edges $v^j_{i,1}, v^j_{i,2}, \ldots, v^j_{i,t}$, and denote this set of vertices by $V_i^j$.
Each $v^j_{i,\gamma}$ \emph{corresponds} to $v_{i,\gamma} \in V_i$.
We call \emph{$j$-th column} the set $\bigcup_{i \in [k]} V_i^j$, and \emph{$i$-th row}, the set $\bigcup_{j \in [m]} V_i^j$.
We set $\mathcal X := \{V_i^j\}_{i \in [k], j \in [m]}$. 
By definition of \gmd, a solution $S$ has to satisfy that for every $i \in [k], j \in [m]$, $|S \cap V_i^j|~=1$.
We call \emph{legal set} a set $S$ of size $k' = km$ that satisfies this property.
We call \emph{consistent set} a legal set $S$ which takes the ``same'' vertex in each row, that is, for every $i \in [k]$, for every pair $(v^j_{i,\gamma},v^{j'}_{i,\gamma'}) \in (S \cap V_i^j) \times (S \cap V_i^{j'})$, then $\gamma = \gamma'$. 

 \subsubsection{Edge gadget}

For each  edge $e_j=v_{i,\gamma} v_{i',\gamma'} \in E(G)$, 
 we add an edge gadget $\mathcal{G}(e_j)$ in the $j$-th column of~$G'$.
$\mathcal{G}(e_j)$ consists of a path on three vertices: $c_jg_jc'_j$.
The pair $\{c_j,c'_j\}$ is added to the list of critical pairs $\mathcal P$.
We link both $v_{i,\gamma}^j$ and $v_{i',\gamma'}^j$ to $g_j$ by a private path\footnote{We use the expression \emph{private path} to emphasize that the different sources get a pairwise internally vertex-disjoint path to the target.} of length $t+2$.
We link the at least two and at most four vertices $v_{i,\gamma-1}^j, v_{i,\gamma+1}^j, v_{i',\gamma'-1}^j, v_{i',\gamma'+1}^j$ (whenever they exist) to $c_j$ by a private path of length $t+2$.
This defines at most six paths from $V_i^j \cup V_{i'}^j$ to $\mathcal{G}(e_j)$.
Let us denote by $W_j$ the at most six endpoints of these paths in $V_i^j \cup V_{i'}^j$.
For each $v \in W_j$, we denote by $P(v,j)$ the path from $v$ to $\mathcal{G}(e_j)$.
We set $E_i^j := \bigcup_{v \in W_j \cap V_i^j} P(v,j)$ and $E_{i'}^j := \bigcup_{v \in W_j \cap V_{i'}^j} P(v,j)$.  
We denote by $X_j$ the set of the at most six neighbors of $W_j$ on the paths to $\mathcal{G}(e_j)$.
Henceforth we may refer to the vertices in some $X_j$ as the \emph{cyan vertices}.
Individually we denote by $e_{i,\gamma}^j$ the cyan vertex neighbor of $v_{i,\gamma}^j$ in $P(v_{i,\gamma}^j,j)$.
We observe that for fixed $i$ and $j$, $e_{i,\gamma}^j$ exists for at most three values of $\gamma$.
We add an edge between two cyan vertices if their respective neighbors in $V_i^j$ are also linked by an edge (or equivalently, if they have consecutive ``indices $\gamma$'').
These extra edges are useless in the \sgmd-instance, but will turn out useful in the \smd-instance.
See~\cref{fig:edgeGadget} for an illustration of the edge gadget.

 

\begin{figure}[h!]
  \centering
  \begin{tikzpicture}[scale=0.8]
     \def\k{3}
    \def\t{5}
    \def\s{0.5}
    \def\x{0.5}
    \def\d{3}
    
    \foreach \i/\c in {1/blue,2/red,3/yellow}{
      \foreach \h in {1,...,\t}{
        \node[draw,circle,fill=\c!20] (v\i\h) at (\d,\t * \k - \i * \k - \i * \x - \s * \h) {} ;
      }
        \node[draw, rectangle, minimum width=0.5cm, minimum height=1cm, rounded corners, very thick, fit=(v\i1) (v\i\t)] (v\i) {} ;
        \node (tv\i\t) at (\d,\t * \k - \i * \k - \i * \x + \x - 0.3) {$V_\i^4$} ; 
    }
    \foreach \i in {1,...,\k}{
      \pgfmathtruncatemacro{\tm}{\t-1}
      \foreach \h in {1,...,\tm}{
        \pgfmathtruncatemacro{\hp}{\h+1}
        \draw (v\i\h) -- (v\i\hp) ;
      }
    }
    \foreach \h in {1,...,\t}{
      \node (tv1\h) at (\d + 0.8,\t * \k - \k - \x - \s * \h) {\footnotesize{$v_{1,\h}^4$}} ;
    }
    \node (te14) at (\d - 1,\t * \k - \k - \x - 3 * \s) {\footnotesize{$e_{1,4}^4$}} ;
    \node (te15) at (\d - 1,\t * \k - \k - \x - 6 * \s) {\footnotesize{$e_{1,5}^4$}} ;

    \foreach \i in {1}{
      \foreach \h in {4,5}{
        \node[draw, circle, fill=cyan] (n\i\h) at (\d - 1,\t * \k - \i * \k - \i * \x - \s * \h) {} ;
      }
    }
    \foreach \i in {3}{
      \foreach \h in {2,3,4}{
        \node[draw, circle, fill=cyan] (n\i\h) at (\d - 1,\t * \k - \i * \k - \i * \x - \s * \h) {} ;
      }
    }
    \draw (v14) -- (n14) -- (n15) -- (v15) ;
    \draw (v32) -- (n32) -- (n33) -- (n34) -- (v34) ;
    \draw (n33) -- (v33) ;
    
    \node[draw, circle] (g4) at (-5,6.6) {} ;
    \node[draw, circle] (c4) at (-5.6,7.2) {} ;
    \node[draw, circle] (cp4) at (-5.6,6) {} ;

    \node (tg4) at (-5.45,6.6) {$g_4$} ;
    \node (tc4) at (-6.2,7.2) {$c_4$} ;
    \node (tcp4) at (-6.2,6) {$c'_4$} ;

    \draw (c4) -- (g4) -- (cp4) ;

    \node[draw, rectangle, thick, rounded corners, fit= (c4) (cp4) (g4)] (ge4) {} ;
    \node (ge4) at (-5.3,5.3) {$\mathcal G(e_4)$} ;

    \foreach \v in {n15,n33}{
      \draw (\v) edge[above] node {6} (g4) ;
    }
    \foreach \v in {n14,n32}{
      \draw (\v) edge[above] node {6} (c4) ;
    }
    \foreach \v in {n34}{
      \draw (\v) edge[below] node {6} (c4) ;
    }
  \end{tikzpicture}
  \caption{The edge gadget $\mathcal G(e_4)$ with $e_4 = v_{1,5}v_{3,3}$.
    Weighted edges are short-hands for subdivisions of the corresponding length.
    The edges between the cyan vertices will not be useful for the \sgmd-instance, but will later simplify the construction of the \smd-instance.}
  \label{fig:edgeGadget}
\end{figure}
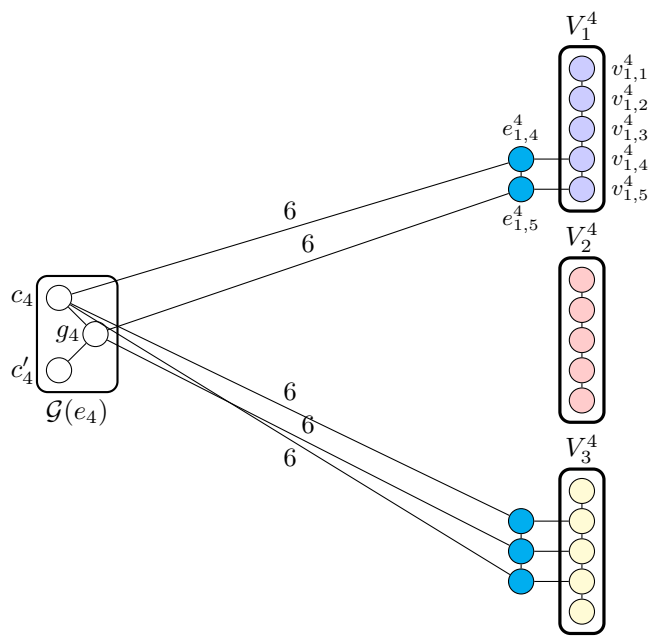

 The rest of the construction will preserve that for every $v \in (V_i^j \cup V_{i'}^j) \setminus \{v_{i,\gamma}^j, v_{i',\gamma'}^j\}$, $\dist(v,c'_j)=\dist(v,c_j)+2$, and for each $v \in \{v_{i,\gamma}^j, v_{i',\gamma'}^j\}$, $\dist(v,c_j) = \dist(v,g_j)+1 = \dist(v,c'_j)$.
 In other words, the only two vertices of $V_i^j \cup V_{i'}^j$ not resolving the critical pair $\{c_j,c'_j\}$ are $v_{i,\gamma}^j$ and $v_{i',\gamma'}^j$, corresponding to the endpoints of $e_j$.

  \subsubsection{Propagation gadget}
  
Between each pair $(V_i^j,V_i^{j+1})$, where $j+1$ is taken modulo $m$, we insert an identical copy of the propagation gadget, and we denote it by $P_i^{j,j+1}$.
It ensures that if the vertex $v_{i,\gamma}^j$ is in a legal resolving set $S$, then the vertex of $S \cap V_i^{j+1}$ should be some $v_{i,\gamma'}^{j+1}$ with $\gamma \leqslant \gamma'$.
The cylindricity of the construction and the fact that exactly one vertex of $V_i^j$ is selected, will therefore impose that the set $S$ is consistent.

$P_{i,}^{j,j+1}$ comprises four vertices \swge, \sege, \nwge, \nege, called \emph{gates}, and a set $A_i^j$ of $2t$ vertices $a_{i,1}^j, \ldots, a_{i,t}^j, \alpha_{i,1}^j, \ldots, \alpha_{i,t}^j$.
We make both $a_{i,1}^ja_{i,2}^j \ldots a_{i,t}^j$ and $\alpha_{i,1}^j\alpha_{i,2}^j \ldots \alpha_{i,t}^j$ a path with $t-1$ edges.
For each $\gamma \in [t]$, we add the pair $\{a_{i,\gamma}^j,\alpha_{i,\gamma}^j\}$ to the set of critical pairs $\mathcal P$.
Removing the gates disconnects $A_i^j$ from the rest of the graph.

We now describe how we link the gates to $V_i^j$, $V_i^{j+1}$, and $A_i^j$.
We link $v_{i,1}^j$ (the ``top'' vertex of $V_i^j$) to \swge and $v_{i,t}^j$ (the ``bottom'' vertex of $V_i^j$) to \nwge both by a path of length~$2$.
We also link $v_{i,1}^{j+1}$ to \sege by a path of length~$3$, and $v_{i,t}^{j+1}$ to \nege by a path of length~$2$.
Then we make \nwge adjacent to $a_{i,1}^j$ and $\alpha_{i,1}^j$, while we make \nege adjacent to $\alpha_{i,1}^j$ only.
We make \sege adjacent to $a_{i,t}^j$ and $\alpha_{i,t}^j$, while we make \swge adjacent to $a_{i,t}^j$ only.
Finally, we add an edge between \nege and \nwge, and between \swge and \sege.
See \cref{fig:propagationExplNew} for an illustration of the propagation gadget $P_i^{j,j+1}$ with $t=5$.

 \begin{figure}[h!]
   \centering
   \begin{tikzpicture}[scale=0.92]
     \def\t{5}
     \def\d{13}
     \def\cd{1.4}
     \foreach \j/\nj in {0/j,1/j+1}{ 
       \foreach \h in {1,...,\t}{
         \node[draw,circle,fill=yellow!20] (v\j\h) at (\d * \j,\t - \h) {} ;
         \node (tv\j\h) at (\d * \j - 0.6 + \j * 1.35,\t - \h) {$v_{i,\h}^{\nj}$} ;
       }
       \node[draw, rectangle, minimum width=0.5cm, minimum height=1cm, rounded corners, very thick, fit=(v\j1) (v\j\t)] (V\j) {} ;
     }
     \node (tV0) at (0,\t - 0.3) {$V_i^j$} ;
     \node (tV1) at (\d,\t - 0.3) {$V_i^{j+1}$} ;

     \node[fill=blue,circle,inner sep=-0.08cm] at (0,\t - 2) {} ;
     \node[fill=red,circle,inner sep=-0.08cm] at (\d,\t - 2) {} ;

     \pgfmathsetmacro{\ce}{\d / 2}
     \foreach \a/\b/\c in {\ce-2/-1.5/sw,\ce+2/-1.5/se,\ce-2/\t+0.5/nw,\ce+2/\t+0.5/ne}{
       \node[draw, rectangle, minimum height=0.5cm, minimum width=1cm] (\c) at (\a, \b) {$\text{\c}_i^j$} ;
     }

     \foreach \h/\dv/\dw/\c in {1/6/7/red,2/7/8/red,3/6/7/blue,4/5/6/blue,5/4/5/blue}{
       \node[fill=\c!20,draw,circle,inner sep=0.08cm] (a\h) at (\d / 2 - \cd,\t - \h) {\footnotesize{\textcolor{blue}{$\dv$}|\textcolor{red}{$\dw$}}} ;
     }
     \foreach \h/\dv/\dw/\c in {1/6/6/red,2/7/7/red,3/7/7/blue,4/6/6/blue,5/5/5/blue}{
       \node[fill=\c!20,draw,circle,inner sep=0.08cm] (alpha\h) at (\d / 2 + \cd,\t - \h) {\footnotesize{\textcolor{blue}{$\dv$}|\textcolor{red}{$\dw$}}} ;
     }
     \foreach \h in {1,...,\t}{
       \node[draw,thin,dashed,rectangle,rounded corners, fit=(a\h) (alpha\h), inner sep=0.03cm] (cp\h) {} ;
     }
     \foreach \h in {1,...,\t}{
       \node (ta\h) at (\d / 2 - \cd - 1,\t - \h) {$a_{i,\h}^j$} ;
       \node (talpha\h) at (\d / 2 + \cd + 1,\t - \h) {$\alpha_{i,\h}^j$} ;
     }
     \foreach \l in {a,alpha,v0,v1}{
       \draw[thick] (\l1) -- (\l2) -- (\l3) -- (\l4) -- (\l5) ;
     }
     \draw (a\t) -- (se) -- (alpha\t) ;
     \draw (a\t) -- (sw) -- (se) ;
     \draw (a1) -- (nw) -- (ne) -- (alpha1) ;
     \draw (a1) -- (nw) -- (alpha1) ;
     \draw (ne) -- (alpha1) ;
     \draw (v0\t) edge[bend left=20] node[above] {$2$} (nw) ;
     \draw (se) edge[bend right=20] node[above] {$3$} (v11) ;
     \draw (v01) edge[bend right=20] node[above] {$2$} (sw) ;
     \draw (v1\t) edge[bend right=20] node[above] {$2$} (ne) ;


   \end{tikzpicture}
   \caption{The propagation gadget $P_i^{j,j+1}$.
     The critical pairs $\{a_{i,\gamma}^j, \alpha_{i,\gamma}^j\}$ are surrounded by thin dashed lines.
     The blue (resp. red) integer on a vertex of $A_i^j$ is its distance to the blue (resp. red) vertex in $V_i^j$ (resp. $V_i^{j+1}$). 
     Note that the blue vertex distinguishes the critical pairs below it, while the red vertex distinguishes critical pairs at its level or above.}
   \label{fig:propagationExplNew}
 \end{figure}

 Let us motivate the gadget $P_i^{j,j+1}$.
 One can observe that the gates \nege and \swge resolve the critical pairs of the propagation gadget, while the gates \nwge and \sege do not.
 Consider that the vertex added to the resolving set in $V_i^j$ is $v_{i,\gamma}^j$.
 Its shortest paths to critical pairs \emph{below} it (that is, with index $\gamma' > \gamma$) go through the gate \swge, whereas its shortest paths to critical pairs at its level or above (that is, with index $\gamma' \leqslant \gamma$) go through the gate \nwge.
Thus $v_{i,\gamma}^j$ only resolves the critical pairs $\{a_{i,\gamma'}^j,\alpha_{i,\gamma'}\}$ with $\gamma' > \gamma$.
On the contrary, the vertex of the resolving set in $V_i^{j+1}$ only resolves the critical pairs $\{a_{i,\gamma'}^j,\alpha_{i,\gamma'}^j\}$ at its level or above.
This will force that its level is $\gamma$ or below.
Hence the vertices of the resolving in $V_i^j$ and $V_i^{j+1}$ should be such that $\gamma' \geqslant \gamma$. 
Since there is also a propagation gadget between $V_i^m$ and $V_i^1$, this circular chain of inequalities forces a global equality.

 \subsubsection{Wrapping up}
 We put the pieces together as described in the previous subsections.
 At this point, it is convenient to give names to the neighbors of $V_i^j$ in the propagation gadgets $P_i^{j-1,j}$ and $P_i^{j,j+1}$.
 We may refer to them as \emph{blue vertices} (as they appear in \cref{fig:propagationAddons}).
 We denote by $\tle$ the neighbor of $v_{i,1}^j$ in $P_i^{j-1,j}$, $\tri$, the neighbor of $v_{i,1}^j$ in $P_i^{j,j+1}$, $\ble$, the neighbor of $v_{i,t}^j$ in $P_i^{j-1,j}$, and $\bri$, the neighbor of $v_{i,t}^j$ in $P_i^{j,j+1}$.
 We add the following edges and paths.
 
 For any pair $i, j$ such that $e_j$ has an endpoint in $V_i$, the vertices $\tle, \tri, \ble, \bri$ are linked to $g_j$ by a private path of length the distance of their unique neighbor in $V_i^j$ to $c_j$.
 We add an edge between $\text{se}_i^j$ and $\text{se}_i^{j+1}$, and between $\text{nw}_i^j$ and $\text{nw}_i^{j+1}$ (where $j+1$ is modulo $m$).
 Finally, for every $e_j \in E(V_i,V_{i'})$, we add four paths between $\text{se}_i^j, \text{se}_{i'}^j, \text{nw}_i^j, \text{nw}_{i'}^j$ and $g_j \in \mathcal G(e_j)$.
 More precisely, for each $i'' \in \{i,i'\}$, we add a path from $g_j$ to $\text{se}_{i''}^j$ of length $\dist(g_j,\text{sw}_{i''}^j)-4$, and a path from $g_j$ to $\text{nw}_{i''}^j$ of length $\dist(g_j,\text{nw}_{i''}^j)-4$.
 These distances are taken in the graph before we introduced the new paths, and one can observe that the length of these paths is at least $t$.
This finishes the construction.
 
We recall that, by a slight abuse of language, a \emph{resolving set} in the context of \gmd is a set which resolves all the critical pairs of $\mathcal P$.
In particular, it is not necessarily a resolving set in the sense of \md.
With that terminology, a solution for \gmd is a legal resolving set.

\subsection{Correctness of the reduction}

We now check that the reduction is correct.
We start with the following technical lemma.
If a set $X$ contains a pair that no vertex of $N(X)$ (that is $N[X] \setminus X$) resolves, then no vertex outside $X$ can distinguish the pair.
\begin{lemma}\label{lem:surrounded}
  Let $X$ be a subset of vertices, and $a, b \in X$ be two distinct vertices.
  If for every vertex $v \in N(X)$, $\dist(v,a) = \dist(v,b)$, then for every vertex $v \notin X$, $\dist(v,a) = \dist(v,b)$.
\end{lemma}
\begin{proof}
  Let $v$ be a vertex outside of $X$.
  We further assume that $v$ is not in $N(X)$, otherwise we can already conclude that it does not distinguish $\{a,b\}$.
  A shortest path from $v$ to $a$, has to go through $N(X)$.
  Let $w_a$ be the first vertex of $N(X)$ met in this shortest path from $v$ to $a$.
  Similarly, let $w_b$ be the first vertex of $N(X)$ met in a shortest path from $v$ to $b$.
  Since $w_a, w_b \in N(X)$, they satisfy $\dist(w_a,a) = \dist(w_a,b)$ and $\dist(w_b,a) = \dist(w_b,b)$.
  Then, $\dist(v,a) \leqslant \dist(v,w_b)+\dist(w_b,a) = \dist(v,w_b) + \dist(w_b,b) = \dist(v,b)$, and $\dist(v,b) \leqslant \dist(v,w_a)+\dist(w_a,b) = \dist(v,w_a) + \dist(w_a,a) = \dist(v,a)$.
  Thus $\dist(v,a) = \dist(v,b)$.
\end{proof}

We use the previous lemma to show that every vertex of a $V_i^j$ only resolves critical pairs in gadgets it is attached to.
This will be useful in the two subsequent lemmas.

\begin{lemma}\label{lem:onlyLocal}
  For any $i \in [k]$, $j \in [m]$, and $v \in V_i^j$, $v$ does not resolve any critical pair outside of $P_i^{j-1,j}$, $P_i^{j,j+1}$ (where indices in the superscript are taken modulo $m$), and $\{c_j,c'_j\}$.
  Furthermore, if $e_j \in E(G)$ has no endpoint in $V_i \subseteq V(G)$, then $v$ does not resolve $\{c_j,c'_j\}$.
\end{lemma}
\begin{proof}
  We first show that $v \in V_i^j$ does not resolve any critical pair in propagation gadgets that are not $P_i^{j-1,j}$ and $P_i^{j,j+1}$.
  Let $\{a_{i',\gamma}^{j'},\alpha_{i',\gamma}^{j'}\}$ be a critical pair in a propagation gadget different from $P_i^{j-1,j}$ and $P_i^{j,j+1}$.
  Let $X$ be the connected component containing $P_{i'}^{j',j'+1}$ of $G' - (\{\text{nw}_{i'}^{j'-1}, \text{se}_{i'}^{j'-1}, \text{nw}_{i'}^{j'+1}, \text{se}_{i'}^{j'+1}\} \cup C_e)$, where $C_e$ comprises $\{c_j',g_j'\}$ if $e_{j'}$ has an endpoint in $V_{i'}$ and $\{c_{j'+1},g_{j'+1}\}$ if $e_{j'+1}$ has an endpoint in $V_{i'}$.
  Thus $C_e$ has size $0$, $2$, or $4$.
  One can observe that $N(X) = \{\text{nw}_{i'}^{j'-1}, \text{se}_{i'}^{j'-1}, \text{nw}_{i'}^{j'+1}, \text{se}_{i'}^{j'+1}\} \cup C_e$, that $V_{i'}^{j'} \cup V_{i'}^{j'+1} \subseteq X$, and that no ``other $V_i^j$'' intersects $X$.
  In particular $V_i^j$ is fully contained in $G - X$.
  We now check that no vertex of $N(X)$ resolves the pair $\{a_{i',\gamma}^{j'},\alpha_{i',\gamma}^{j'}\}$ (which is inside $X$).
  For each $u \in \{\text{nw}_{i'}^{j'-1}, \text{nw}_{i'}^{j'+1}\}$, it holds that $\dist(u,a_{i',\gamma}^{j'}) = \gamma+1 = \dist(u,a_{i',\gamma}^{j'})$ (the shortest paths go through $\text{nw}_{i'}^{j'}$), while for each $u \in \{\text{se}_{i'}^{j'-1}, \text{se}_{i'}^{j'+1}$, it holds that $\dist(u,a_{i',\gamma}^{j'}) = t - \gamma + 2 = \dist(u,a_{i',\gamma}^{j'})$ (the shortest paths go through $\text{se}_{i'}^{j'}$).
  If they are part of $C_e$, $g_{j'}$ and $c_{j'}$ also do not resolve $\{a_{i',\gamma}^{j'},\alpha_{i',\gamma}^{j'}\}$, the shortest paths going through the gates $\text{nw}_{i'}^{j'}$ or $\text{se}_{i'}^{j'}$, and respectively $g_j$ and then the gates $\text{nw}_{i'}^{j'}$ or $\text{se}_{i'}^{j'}$.
  For the same reason, $g_{j'+1}$ and $c_{j'+1}$ do not resolve $\{a_{i',\gamma}^{j'},\alpha_{i',\gamma}^{j'}\}$.
  Then we conclude by \cref{lem:surrounded} that no vertex of $V_i^j$ (in particular $v$) resolves $\{a_{i',\gamma}^{j'},\alpha_{i',\gamma}^{j'}\}$, or any critical pair in $P_{i'}^{j'}$.
    
  Let us now show that the pair $\{c_j,c'_j\}$ is not resolved by any vertex of $\cup \mathcal X \setminus (V_{i'}^j \cup V_{i''}^j)$ such that $e_j \in E(V_{i'},V_{i''})$.
  Let $Y := \{\text{tl}_{i'}^j, \text{tr}_{i'}^j, \text{bl}_{i'}^j, \text{br}_{i'}^j, \text{tl}_{i''}^j, \text{tr}_{i''}^j, \text{bl}_{i''}^j, \text{br}_{i''}^j, \text{nw}_{i'}^j, \text{se}_{i'}^j, \text{nw}_{i''}^j, \text{se}_{i''}^j\}$, and $X$ be the connected component containing $g_j$ in $G' - Y$.
  Again one can observe that $N(X) = Y$, $X$ contains $V_{i'}^j \cup V_{i''}^j$ but does not intersect any ``other $V_i^j$''.
  We therefore show that no vertex of $Y$ resolves $\{c_j,c'_j\}$, and conclude with \cref{lem:surrounded}.
  All the vertices of $Y$ have a private path to $g_j$ whose length is such that they have a shortest path to $c_j$ going through~$g_j$.
  Therefore $\forall u \in Y$, $\dist(u,c_j) = \dist(u,g_j)+1 = \dist(u,c'_j)$. 
\end{proof}

The two following lemmas show the equivalences relative to the expected use of the edge and propagation gadgets.
They will be useful in \cref{subsubsec:misToLRS,subsubsec:LRSToMis}.

\begin{lemma}\label{lem:edge-check}
  A legal set $S$ resolves the critical pair $\{c_j,c'_j\}$ with $e_j = v_{i,\gamma}v_{i',\gamma'}$ if and only if the vertex $v_{i,\gamma_i}^j$ in $V_i^j \cap S$ and the vertex $v_{i',\gamma_{i'}}^j$ in $V_{i'}^j \cap S$ satisfy $(\gamma,\gamma') \neq (\gamma_i,\gamma_{i'})$.
\end{lemma}
\begin{proof}
  By \cref{lem:onlyLocal}, no vertex of $S \setminus \{v_{i,\gamma_i}^j,v_{i',\gamma_{i'}}^j\}$ resolves $\{c_j,c'_j\}$.
  By construction of $G'$, $v_{i,\gamma}^j$ (resp. $v_{i',\gamma'}^j$) is the only vertex of $V_i^j$ (resp. $V_{i'}^j$) that does not resolve $\{c_j,c'_j\}$.
  Indeed the shortest paths of $v_{i,\gamma''}^j$, for $\gamma'' \geqslant \gamma+1$, to $\{c_j,c'_j\}$ go through $v_{i,\gamma+1}^j$ which resolves the pair.
  Note that a shortest path between $V_i^j$ and $V_{i'}^j$ has length at least $2t+4$, so a shortest path from $v_{i,\gamma''}^j$ to $\{c_j,c'_j\}$ cannot go through $V_{i'}^j$.
  Similarly the shortest paths of $v_{i,\gamma''}^j$, for $\gamma'' \leqslant \gamma-1$, to $\{c_j,c'_j\}$ go through $v_{i,\gamma-1}^j$ which also resolves the pair.
  Thus only $v_{i,\gamma}^j$ (resp.~$v_{i',\gamma'}^j$), whose shortest paths to $\{c_j,c'_j\}$ go via $g_j$, does not resolve this pair among $V_i^j$ (resp. $V_{i'}^j$).
  Hence, the critical pair $\{c_j,c'_j\}$ is not resolved by $S$ if and only if $v_{i,\gamma_i}^j = v_{i,\gamma}^j$ and $v_{i',\gamma_{i'}}^j = v_{i',\gamma'}^j$.
\end{proof}

\begin{lemma}\label{lem:propagation-check}
  A legal set $S$ resolves all the critical pairs of $P_i^{j,j+1}$ if and only if the vertex $v_{i,\gamma}^j$ in $V_i^j \cap S$ and the vertex $v_{i,\gamma'}^{j+1}$ in $V_i^{j+1} \cap S$ satisfy $\gamma \leqslant \gamma'$.
\end{lemma}
\begin{proof}
By \cref{lem:onlyLocal}, no vertex of $S \setminus \{v_{i,\gamma}^j,v_{i,\gamma'}^{j+1}\}$ resolves a critical pair of $P_i^{j,j+1}$.
  Let us show that the critical pairs that $v_{i,\gamma}^j$ resolves in $A_i^j$ are exactly the pairs $\{a_{i,z}^j,\alpha_{i,z}^j\}$ with $z > \gamma$.
  For any $z \in [t]$, it holds that $\dist(v_{i,\gamma}^j,a_{i,z}^j) = \min(t+2+z-\gamma,t+2+\gamma-z) = t+2+\min(z-\gamma,\gamma-z)$,
  and $\dist(v_{i,\gamma}^j,\alpha_{i,z}^j) = \min(t+2+z-\gamma,t+3+\gamma-z) = t+2+\min(z-\gamma,\gamma-z+1)$.
  So if $z > \gamma$, $\dist(v_{i,\gamma}^j,a_{i,z}^j) = t + 2 + \gamma - z \neq t + 2 + \gamma - z + 1 = \dist(v_{i,\gamma}^j,\alpha_{i,z}^j)$.
  Whereas if $z \leqslant \gamma$, $\dist(v_{i,\gamma}^j,a_{i,z}^j) = t+2+z-\gamma = \dist(v_{i,\gamma}^j,\alpha_{i,z}^j)$.
  
  Similarly, we show that the critical pairs that $v_{i,\gamma'}^{j+1}$ resolves in $A_i^j$ are exactly the pairs $\{a_{i,z}^j,\alpha_{i,z}^j\}$ with $z \leqslant \gamma'$.
  For every $z \in [t]$, it holds that $\dist(v_{i,\gamma'}^{j+1},a_{i,z}^j) = \min(t+3+z-\gamma',t+3+\gamma'-z) = t+3+\min(z-\gamma',\gamma'-z)$, and $\dist(v_{i,\gamma'}^{j+1},\alpha_{i,z}^j) = \min(t+2+z-\gamma',t+3+\gamma'-z) = t+2+\min(z-\gamma',\gamma'-z+1)$.
  So if $z \leqslant \gamma'$, $\dist(v_{i,\gamma'}^{j+1},a_{i,z}^j) = t + 3 + z -\gamma' \neq t + 2 + z - \gamma'= \dist(v_{i,\gamma'}^{j+1},\alpha_{i,z}^j)$.
  Whereas if $z > \gamma'$, $\dist(v_{i,\gamma'}^{j+1},a_{i,z}^j) = t + 3 + \gamma' - z = \dist(v_{i,\gamma'}^{j+1},\alpha_{i,z}^j)$.
    This implies that all the critical pairs of $A_i^j$ are resolved by $S$ if and only if $\gamma \leqslant \gamma'$.
\end{proof}

We can now prove the correctness of the reduction.
The construction can be computed in polynomial time in $|V(G)|$, and $G'$ itself has size bounded by a polynomial in $|V(G)|$.
We postpone checking that the pathwidth is bounded by $O(k)$ to the end of the second step, where we produce an instance of \smd whose graph $G''$ admits $G'$ as an induced subgraph.

\subsubsection{$k$-Multicolored Independent Set in $G$ $\Rightarrow$ legal resolving set in $G'$.}\label{subsubsec:misToLRS}

Let $\{v_{1,\gamma_1}, \ldots, v_{k,\gamma_k}\}$ be a $k$-multicolored independent set in $G$.
We claim that $S := \bigcup_{j \in [m]} \{v_{1,\gamma_1}^j, \ldots, v_{k,\gamma_k}^j\}$ is a legal resolving set in $G'$ (of size $km$).
The set $S$ is legal by construction.
Since for every $i \in [k]$, and $j \in [m]$, $v_{i,\gamma_i}^j$ and $v_{i,\gamma_i}^{j+1}$ are in $S$ ($j+1$ is modulo $m$), all the critical pairs in the propagation gadgets are resolved by $S$, by \cref{lem:propagation-check}.
Since $\{v_{1,\gamma_1}, \ldots, v_{k,\gamma_k}\}$ is an independent set in $G$, there is no $e_j = v_{i,\gamma}v_{i',\gamma'} \in E(G)$, such that $(\gamma,\gamma')=(\gamma_i,\gamma_{i'})$.
Thus every critical pair $\{c_j,c'_j\}$ is resolved by $S$, by \cref{lem:edge-check}.

\subsubsection{Legal resolving set in $G'$ $\Rightarrow$ $k$-Multicolored Independent Set in $G$.}\label{subsubsec:LRSToMis}

Assume that there is a legal resolving set $S$ in $G'$.
For every $i \in [k]$, for every $j \in [m]$, the vertex $v_{i,\gamma(i,j)}^j$ in $V_i^j \cap S$ and the vertex $v_{i,\gamma(i,j+1)}^{j+1}$ in $V_i^{j+1} \cap S$ ($j+1$ is modulo $m$) are such that $\gamma(i,j) \leqslant \gamma(i,j+1)$, by \cref{lem:propagation-check}.
Thus $\gamma(i,1) \leqslant \gamma(i,2) \leqslant \ldots \leqslant \gamma(i,m-1) \leqslant \gamma(i,m) \leqslant \gamma(i,1)$, and $\gamma_i := \gamma(i,1) = \gamma(i,2) = \ldots = \gamma(i,m-1) = \gamma(i,m)$.
We claim that $\{v_{1,\gamma_1}, \ldots, v_{k,\gamma_k}\}$ is a $k$-multicolored independent set in $G$.
Indeed, there cannot be an edge $e_j = v_{i,\gamma_i} v_{i',\gamma_{i'}} \in E(G)$, since otherwise the critical pair $\{c_j,c'_j\}$ is not resolved, by \cref{lem:edge-check}.

\section{Parameterized hardness of \md/$\tw$}\label{sec:md}

In this section, we produce in polynomial time an instance $(G'',k'')$ of \md equivalent to $(G',\mathcal{X},km,\mathcal P)$ of \gmd.
The graph $G''$ has also pathwidth $O(k)$.
Now, an instance is just a graph and an integer.
There is no longer $\mathcal X$ and $\mathcal P$ to constrain and respectively loosen the ``resolving set'' at our convenience.
This creates two issues: $(1)$ the vertices outside the former set $\mathcal X$ can now be put in the resolving set, potentially yielding undesired solutions\footnote{Also, it is now possible to put two or more vertices of the same $V_i^j$ in the resolving set $S$} and $(2)$ our candidate solution (when there is a $k$-multicolored independent set in $G$) may not distinguish all the vertices.

\subsection{Construction}

We settle both issues by attaching new gadgets to $G'$.
Eventually the new graph $G''$ will contain $G'$ as an induced subgraph. 
To settle the issue (1), we design a \emph{forced set} gadget.
A forced set gadget attached to $V_i^j$ contains two pairs of vertices which are only resolved by vertices of $V_i^j$.
Thus the gadget simulates the action of $\mathcal X$.

There are a few pairs which are not resolved by a solution of \gmd.
To make sure that all pairs are resolved, we add vertices which need be selected in the resolving set.
Technically we could use the previous gadget on a singleton set.
But we can make it simpler: we just attach two pendant neighbors, that we then make adjacent, to some chosen vertices.
A pair of pendant neighbors are false twins in the whole graph.
So we know that at least one of these two vertices have to be in the resolving set.
Hence we call that the \emph{forced vertex} gadget, and one of the false twins, \emph{a forced vertex}.
It is important that these forced vertices do not resolve any pair of $\mathcal P$.
So we can only add pendant twins to vertices themselves not resolving any pair of $\mathcal P$. 

\subsubsection{Forced set gadget}\label{sec:forcedSet}
To deal with the issue (1), we introduce two new pairs of vertices for each $V_i^j$.
The intention is that the only vertices resolving both these pairs simultaneously are precisely the vertices of $V_i^j$.
For any $i \in [k]$ and $j \in [m]$, we add to $G'$ two pairs of vertices $\{p_i^j,q_i^j\}$ and $\{r_i^j,s_i^j\}$, and two \emph{gates} $\pi_i^j$ and $\rho_i^j$. 
Vertex $\pi_i^j$ is adjacent to $p_i^j$ and $q_i^j$, and vertex $\rho_i^j$ is adjacent to $r_i^j$ and $s_i^j$.

We link $v_{i,1}^j$ to $p_i^j$, and $v_{i,t}^j$ to $r_i^j$, each by a path of length~$t$.
It introduces two new neighbors of $v_{i,1}^j$ and $v_{i,t}^j$ (the brown vertices in \cref{fig:propagationAddons}).
We denote them by $\tbr$ and $\bbr$, respectively.
The blue and brown vertices are linked to $\pi_i^j$ and $\rho_i^j$ in the following way.
We link $\tle$ and $\tri$ to $\pi_i^j$ by a private path of length $t$, and to $\rho_i^j$ by a private path of length~$2t-1$.
We link $\ble$ and $\bri$ to $\pi_i^j$ by a private path of length $2t-1$, and to $\rho_i^j$ by a private path of length~$t$.
(Let us clarify that the names of the blue vertices $\ble$ and $\bri$ are for ``bottom-left'' and ``bottom-right'', and \emph{not} for ``blue'' and ``brown''.) 
We link $\tbr$ (neighbor of $v_{i,1}^j$) to $\rho_i^j$ by a private path of length~$2t-1$.
We link $\bbr$ (neighbor of $v_{i,t}^j$) to $\pi_i^j$ by a private path of length~$2t-1$.
Note that the general rule to set the path length is to match the distance between the neighbor in $V_i^j$ and $p_i^j$ (resp.~$r_i^j$).
With that in mind we link, if it exists, the \emph{top cyan vertex $\tc$} (the one with smallest index $\gamma$) neighboring $V_i^j$ to $\pi_i^j$ with a path of length $\dist(v_{i,\gamma}^j,p_i^j) = t+\gamma-1$ where $v_{i,\gamma}^j$ is the unique vertex in $N(\tc) \cap V_i^j$.
Observe that with the notations of the previous section $\tc = e_{i,\gamma}^j$.
We also link, if it exists, the \emph{bottom cyan vertex $\bc$} (the one with largest index $\gamma$) to $\rho_i^j$ with a path of length $\dist(v,r_i^j)$ where $v$ is again the unique neighbor of $\bc$ in $V_i^j$. 

It can be observed that we only have two paths (and not all six) from the at most three cyan vertices to the gates $\pi_i^j$ and $\rho_i^j$.
This is where the edges between the cyan vertices will become relevant.
See \cref{fig:propagationAddons} for an illustration of the forced vertex gadget, keeping in mind that, for the sake of legibility, four paths to $\{\pi_i^j, \rho_i^j\}$ are not represented.

 \subsubsection{Forced vertex gadget}\label{sec:forcedVertex}
 We now deal with the issue (2).
 By \emph{we add} (or \emph{attach}) \emph{a forced vertex} to an already present vertex $v$, we mean that we add two adjacent neighbors to $v$, and that these two vertices remain of degree 2 in the whole graph $G''$.
 Hence one of the two neighbors will have to be selected in the resolving set since they are false twins.
 We call \emph{forced vertex} one of these two vertices (picking arbitrarily).
 
 For every $i \in [k]$ and $j \in [m]$, we add a forced vertex to the gates \nwge and \sege of $P_i^{j,{j+1}}$.
 We also add a forced vertex to each vertex in $N(\{\pi_i^j,\rho_i^j\}) \setminus \{p_i^j, q_i^j, r_i^j, s_i^j\}$.
 This represents a total of $12$ vertices ($6$ neighbors of $\pi_i^j$ and $6$ neighbors of $\rho_i^j$). 
 For every $j \in [m]$, we attach a forced vertex to each vertex in $N(g_j) \setminus \{c_j, c'_j\}$.
 This constitutes $14$ neighbors (hence $14$ new forced vertices).
 Therefore we set $k'' := km + 12km + 2km + 14m = 15km + 14m$.

 \begin{figure}[h!]
   \centering
   \begin{tikzpicture}[scale=0.85]
     \def\t{5}
     \def\d{10}
     \def\cd{0.7}
     \def\opa{0.35}
     \foreach \h in {1,...,\t}{
       \node[draw,circle,fill=yellow!20] (v\h) at (0,\t - \h) {} ;
     }
     \node[draw, rectangle, minimum width=0.5cm, minimum height=1cm, rounded corners, very thick, fit=(v1) (v\t)] (V) {} ;
     \node (tV) at (0,\t) {$V_i^j$} ;

     \foreach \h in {2,3,4}{
       \node[draw,circle,fill=cyan] (w\h) at (-0.75,\t - \h) {} ;
     }
     \draw (w2) -- (w3) -- (w4) ;


     \foreach \i in {2,3,4}{
       \pgfmathtruncatemacro{\im}{\i - 1}
       \draw (v\im) -- (w\i) ;
     }
     \node at (-3,-1.7) {$\mathcal G(e_j)$} ;
     \node[draw, circle] (cpj) at (-3.5,-1) {} ;
     \node[draw, circle] (cj) at (-2.5,-1) {} ;
     \node[draw, circle, fill=red] (gj) at (-3,-0.33) {} ;
     \node (tgj) at (-3.4,-0.33) {$g_j$} ;
     \node (tcj) at (-1.9,-1) {$c_j$} ;
     \node (tcpj) at (-4.1,-1) {$c'_j$} ;
     
     \draw (cpj) -- (gj) -- (cj) ;
     \node[draw, rectangle, rounded corners, thick, fit=(cj) (cpj) (gj)] (gej) {} ;

     \draw (w3) edge[bend right=5] node[left] {$6$} (gj) ;
     \draw (w4) edge node[below] {$6$} (cj) ;
     \draw (w2) edge[bend left=5] node[below] {$6$} (cj) ;
     
     \foreach \l/\x/\y in {z1/-0.65/\t-1, z2/0.65/\t-1, z3/-0.65/0, z4/0.65/0}{
     \node[draw,circle,fill=blue] (\l) at (\x,\y) {} ;
     }
     

     \pgfmathsetmacro{\ce}{\d / 2}
     \foreach \a/\b/\c in {\ce+0.5/-1.5/sw,\ce+2/-1.5/se,\ce+0.5/\t+0.5/nw,\ce+2/\t+0.5/ne}{
       \node[draw, rectangle, minimum height=0.5cm, minimum width=1cm] (\c) at (\a, \b) {$\text{\c}_i^j$} ;
     }

     \foreach \h in {1,...,\t}{
       \node[draw,circle] (a\h) at (\d / 2 - \cd + 1.25,\t - \h) {} ;
     }
     \foreach \h in {1,...,\t}{
       \node[draw,circle] (alpha\h) at (\d / 2 + \cd +1.25,\t - \h) {} ;
     }

     \begin{scope}[xshift=- \d cm]
       \pgfmathsetmacro{\ce}{\d / 2}
     \foreach \a/\b/\c/\lc in {\ce-2.5/-1.5/bsw/sw,\ce-1/-1.5/bse/se,\ce-2.5/\t+0.5/bnw/nw,\ce-1/\t+0.5/bne/ne}{
       \node[draw, rectangle, minimum height=0.5cm, minimum width=1cm] (\c) at (\a, \b) {$\text{\lc}_i^{j-1}$} ;
     }

     \foreach \h in {1,...,\t}{
       \node[draw,circle] (ba\h) at (\d / 2 - \cd-1.75,\t - \h) {} ;
     }
     \foreach \h in {1,...,\t}{
       \node[draw,circle] (balpha\h) at (\d / 2 + \cd-1.75,\t - \h) {} ;
     }
     \end{scope}
     
     \foreach \l in {a,alpha,ba,balpha}{
       \draw[opacity=\opa] (\l1) -- (\l2) -- (\l3) -- (\l4) -- (\l5) ;
     }
     \draw (v1) -- (v2) -- (v3) -- (v4) -- (v5) ;
     \draw[opacity=\opa] (a\t) -- (se) -- (alpha\t) ;
     \draw[opacity=\opa] (a\t) -- (sw) -- (se) ;
     \draw[opacity=\opa] (a1) -- (nw) -- (ne) -- (alpha1) ;
     \draw[opacity=\opa] (a1) -- (nw) -- (alpha1) ;
     \draw[opacity=\opa] (ne) -- (alpha1) ;

     \draw[opacity=\opa] (ba\t) -- (bse) -- (balpha\t) ;
     \draw[opacity=\opa] (ba\t) -- (bsw) -- (bse) ;
     \draw[opacity=\opa] (ba1) -- (bnw) -- (bne) -- (balpha1) ;
     \draw[opacity=\opa] (ba1) -- (bnw) -- (balpha1) ;
     \draw[opacity=\opa] (bne) -- (balpha1) ;

     \draw[opacity=\opa] (z1) -- (v1) -- (z2) ;
     \draw[opacity=\opa] (z3) -- (v\t) -- (z4) ;
     
     \draw[opacity=\opa] (z1) edge node[left] {$2$} (bse) ;
     \draw[opacity=\opa] (z3) edge node[left] {} (bne) ;
     \draw[opacity=\opa] (z2) edge node[right] {} (sw) ;
     \draw[opacity=\opa] (z4) edge node[right] {} (nw) ;





     \foreach \j/\lj in {-0.75/1,\t-0.25/2}{
       \node[draw,circle,fill=brown] (d1\lj) at (1,\j) {} ;
       \node[draw,circle] (d2\lj) at (3,\j) {} ;
       \node[draw,circle,fill=red] (d3\lj) at (3.5,\j) {} ;
       \node[draw,circle] (d4\lj) at (4,\j) {} ;

       \draw (d2\lj) -- (d3\lj) -- (d4\lj) ;
       \draw (d1\lj) edge node[above] {$4$} (d2\lj) ;
     }
     \node (pij) at (3,\t + 0.25) {$p_i^j$} ;
     \node (pi) at (3.5,\t + 0.25) {$\pi_i^j$} ;
     \node (qij) at (4,\t + 0.25) {$q_i^j$} ;

     \node (rij) at (3,-1.25) {$r_i^j$} ;
     \node (rho) at (3.5,-1.25) {$\rho_i^j$} ;
     \node (sij) at (4,-1.25) {$s_i^j$} ;
     
     
     \draw (v1) -- (d12) ;
     \draw (v\t) -- (d11) ;
     \draw (z2) edge[bend right=15] node[below] {$5$} (d32) ;
     \draw (z2) edge node[below] {$9$} (d31) ;
     \draw (z4) edge[bend left=15] node[above] {$5$} (d31) ;
     \draw (z4) edge node[above] {$9$} (d32) ;

     \draw (z1) edge[bend left=40] node[above] {$5$} (d32) ;
     \draw (z3) edge[bend right=40] node[below] {$5$} (d31) ;

     \draw (gj) edge[bend left=70] node[above] {$8$} (nw) ;
     \draw (gj) edge[bend right=30] node[above] {$6$} (se) ;

     

     \draw (d12) edge[bend right=25] node[left] {$9$} (d31) ;
     \draw (d11) edge[bend left=25] node[left] {$9$} (d32) ;

     \node[draw, circle, fill=black] (dumbl1) at (-6.25,-2.5) {} ;
     \node[draw, circle] (dumbl2) at (-5.75,-2.5) {} ;

     \draw (dumbl1) -- (bse) -- (dumbl2) -- (dumbl1) ;

     \node[draw, circle, fill=black] (dumbr1) at (6.75,-2.5) {} ;
     \node[draw, circle] (dumbr2) at (7.25,-2.5) {} ;

     \draw (dumbr1) -- (se) -- (dumbr2) -- (dumbr1) ;

     \node[draw,circle,fill=black] (dumtr1) at (4.5,\t + 0.75) {} ;
     \node[draw,circle] (dumtr2) at (4.5,\t + 0.25) {} ;
     \draw (dumtr1) -- (nw) -- (dumtr2) -- (dumtr1);

     \node[draw,circle,fill=black] (dumtl1) at (-8.6,\t + 0.75) {} ;
     \node[draw,circle] (dumtl2) at (-8.6,\t + 0.25) {} ;
     \draw (dumtl1) -- (bnw) -- (dumtl2) -- (dumtl1) ;

     \draw (se) edge[bend left=25] node[above] {} (bse) ;
     \draw (nw.north) edge[bend right=20] node[above] {} (bnw.north) ;
   \end{tikzpicture}
   \caption{Vertices $\tle, \tri, \ble, \bri$ (blue vertices) are linked to $\pi_i^j$, $\rho_i^j$ by paths of appropriate lengths (see \cref{sec:forcedSet}).
     Vertex $\tbr$ is linked by a path to $\rho_i^j$, while $\bbr$ is linked by a path to $\pi_i^j$.
     To avoid cluttering the figure, we did not represent four paths: from $\tle$ and $\bc$ to $\rho_i^j$, and from $\ble$ and $\tc$ to $\pi_i^j$.
     We also did not represent the paths already in the \sgmd-instance from the blue vertices to $g_j$.
     Black vertices are forced vertices.
     Gray edges are the edges in the propagation gadgets already depicted in Figure~\ref{fig:propagationExplNew}.
     Not represented on the figure, we add a forced vertex to each neighbor of the red vertices, except $p_i^j, q_i^j, r_i^j, s_i^j, c_j, c'_j$.
     Finally we add four more paths and potentially two edges (see \cref{subsec:finalTouches}).}
   \label{fig:propagationAddons}
  \end{figure}

 \subsubsection{Finishing touches and useful notations}\label{subsec:finalTouches}
 We use the convention that $P(u,v)$ denotes the path from $u$ to $v$ which was specifically built from $u$ to $v$.
 In other words, for $P(u,v)$ to make sense, there should be a point in the construction where we say that we add a (private) path between $u$ and $v$.
 For the sake of legibility, $P(u,v)$ may denote either the set of vertices or the induced subgraph.
 We also denote by $\nu(u,v)$ the neighbor of $u$ in the path $P(u,v)$.
 Observe that $P(u,v)$ is a symmetric notation but not $\nu(u,v)$. 

 We add a path of length $\dist(\nu(\pi_i^j,\tri),\sw)=t$ between $\nu(\pi_i^j,\tri)$ and $\se$, and a path of length $\dist(\nu(\pi_i^j,\ble),\text{ne}_i^{j-1})=2t-1$ between $\nu(\pi_i^j,\ble)$ and $\text{nw}_i^{j-1}$.
 Similarly, we add a path of length $\dist(\nu(\rho_i^j,\tri),\sw)=2t-1$ between $\nu(\rho_i^j,\tri)$ and $\se$, and a path of length $\dist(\nu(\rho_i^j,\ble),\text{ne}_i^{j-1})=t$ between $\nu(\rho_i^j,\ble)$ and $\text{nw}_i^{j-1}$.
 We added these four paths so that no forced vertex resolves any critical pair in the propagation gadgets $P_i^{j-1,j}$ and $P_i^{j,j+1}$.

 Finally we add an edge between $\nu(g_j,\nw)$ and $\nu(c_j,\bc)$ whenever $V_i^j$ have exactly three cyan vertices.
 We do that to resolve the pair $\{\nu(c_j,\tc), \nu(c_j,\bc)\}$, and more generally every pair $\{x,y\} \in P(c_j,\tc) \times P(c_j,\bc)$ such that $\dist(c_j,x)=\dist(c_j,y)$.
 This finishes the construction of the instance $(G'',k'' := 15km+14m)$ of \md.
 
\subsection{Correctness of the reduction} 

The two next lemmas will be crucial in \cref{sec:correctnessMDtoGMD}.
The first lemma shows how the forcing set gadget simulates the action of former set $\mathcal X$.

\begin{lemma}\label{lem:tooFar-useless}
  For every $i \in [k]$ and $j \in [m]$,
  \begin{itemize}
  \item $\forall v \in V_i^j$, $v$ resolves both pairs $\{p_i^j, q_i^j\}$ and $\{r_i^j, s_i^j\}$,
  \item $\forall v \notin V_i^j$, $v$ resolves at most one pair of $\{p_i^j, q_i^j\}$ and $\{r_i^j, s_i^j\}$,
  \item $\forall v \notin V_i^j \cup P(v_{i,1}^j,p_i^j) \cup P(v_{i,t}^j,r_i^j) \cup \{q_i^j, s_i^j\}$, $v$ does not resolve $\{p_i^j, q_i^j\}$ nor $\{r_i^j, s_i^j\}$.
  \end{itemize}
\end{lemma}
\begin{proof}
  Let $Y := \{\tle,\tri,\ble,\bri\} \cup (X_j \cap N(V_i^j)) \cup (N(\{\pi_i^j,\rho_i^j\}) \setminus \{p_i^j, q_i^j, r_i^j, s_i^j\})$, and recall that $X_j \cap N(V_i^j)$ is the set of cyan vertices neighbors of $V_i^j$ (if they exist).
  Let us assume that these cyan vertices exist (otherwise the proof is just simpler).
  In particular, there are at least two cyan neighbors $\tc, \bc \in X_j \cap N(V_i^j)$.
  Let $X$ be the connected component of $G-Y$ containing $\{\pi_i^j,\rho_i^j\}$.
  For every vertex $u \in \{\tle,\tri,\ble,\bri,\tc,\bc\}$, by the way we chose the length of $P(u,\pi_i^j)$ (resp.~$P(u,\rho_i^j)$), there is a shortest path from $u$ to $p_i^j$ (resp.~$r_i^j$) that goes through $\pi_i^j$ (resp.~$\rho_i^j$).
  Thus $\dist(u,p_i^j)=\dist(u,\pi_i^j)+1=\dist(u,q_i^j)$ and $\dist(u,r_i^j)=\dist(u,\rho_i^j)+1=\dist(u,s_i^j)$.

  Let $\text{mc}_i^j$ be the middle cyan vertex if it exists (the one which is not the top nor the bottom one).
  There is shortest path from $\text{mc}_i^j$ to $p_i^j$ (resp.~$r_i^j$) going via $\tc$ (resp.~$\bc$) and then $\pi_i^j$ (resp.~$\rho_i^j$).
  This is where the edges $\text{mc}_i^j\tc$ and $\text{mc}_i^j\bc$ are useful.
  Hence $\text{mc}_i^j$ does not resolve $\{p_i^j, q_i^j\}$ nor $\{r_i^j, s_i^j\}$, either.
  It is direct that no vertex of $N(\{\pi_i^j,\rho_i^j\}) \setminus \{p_i^j, q_i^j, r_i^j, s_i^j\}$ resolves $\{p_i^j, q_i^j\}$ nor $\{r_i^j, s_i^j\}$.
  Thus no vertex of $Y$ resolves any of $\{p_i^j, q_i^j\}$ and $\{r_i^j, s_i^j\}$.
  Therefore by \cref{lem:surrounded}, no vertex outside $X$ resolves any of $\{p_i^j, q_i^j\}$ and $\{r_i^j, s_i^j\}$.

  We observe that $X = V_i^j \cup P(v_{i,1}^j,p_i^j) \cup  P(v_{i,t}^j,r_i^j) \cup \{\pi_i^j, q_i^j, \rho_i^j, s_i^j\}$.
  Because of the path from the top brown vertex to $\rho_i^j$, vertices of $P(v_{i,1}^j,p_i^j) \setminus \{v_{i,1}^j\} \cup \{q_i^j\}$, which do resolve $\{p_i^j, q_i^j\}$, do \emph{not} resolve $\{r_i^j, s_i^j\}$.
  Similarly because of the path from the bottom brown vertex to $\pi_i^j$, vertices of $P(v_{i,t}^j,r_i^j) \setminus \{v_{i,t}^j\} \cup \{s_i^j\}$, which do resolve $\{r_i^j, s_i^j\}$, do \emph{not} resolve $\{p_i^j, q_i^j\}$.
  Finally for every $u \in V_i^j$, $\dist(u,q_i^j) = \dist(u,p_i^j)+2$ and $\dist(u,r_i^j) = \dist(u,s_i^j)+2$.
  Therefore vertices of $V_i^j$ are the only ones resolving both $\{p_i^j, q_i^j\}$ and $\{r_i^j, s_i^j\}$, while no vertex of $G-X$ resolves any of these pairs.
\end{proof}

We denote by $f(v)$ the forced vertex attached to a vertex $v$.
For \cref{sec:correctnessMDtoGMD}, we also need the following lemma, which states that the forced vertices do not resolve critical pairs.

\begin{lemma}\label{lem:forcedAndCritical}
  No forced vertex resolves a pair of $\mathcal P$.
\end{lemma}
\begin{proof}
  We first show that no critical pair in some $P_i^{j,j+1}$ is resolved by a forced vertex.
  We use a similar plan as for the proof of \cref{lem:onlyLocal}.
  Let $Y := \{\text{nw}_i^{j-1}, \text{se}_i^{j-1}, \text{nw}_i^{j+1}, \text{se}_i^{j+1}\} \cup C_e$, where $C_e$ comprises $\{c_j,g_j\}$ if $e_j$ has an endpoint in $V_i$ and $\{c_{j+1},g_{j+1}\}$ if $e_{j+1}$ has an endpoint in $V_i$.
  Let $X$ be the connected component of $G'' - Y$ containing $P_i^{j,j+1}$.
  Note that the distances between the vertices of $Y$ and the critical pairs in $P_i^{j,j+1}$ are the same between $G'$ and $G''$.
  Hence as we showed in \cref{lem:onlyLocal}, no vertex of $Y$ resolves a critical pair in $P_i^{j,j+1}$.
  Thus by \cref{lem:surrounded} no vertex outside $X$ resolves a critical pair in $P_i^{j,j+1}$.

  We now check that no forced vertex in $X$ resolves a critical pair in $P_i^{j,j+1}$.
  We show that every forced vertex in $X$ has a shortest path to $\{\nw,\nee\}$ ending in $\nw$, and a shortest path to $\{\sw,\se\}$ ending in $\se$.
  It is clear for $f(\nw)$ and for $f(\se)$, as well as for all the forced vertices attached to neighbors of $g_j$ (in case $e_j$ has an endpoint in $V_i$).
  Indeed recall that the length of $P(g_j,\nw)$ (resp.~$P(g_j,\se)$) is four less than the distance to $\nw$ (resp.~$\sw$) ignoring the path $P(g_j,\nw)$ (resp.~$P(g_j,\se)$).
  So the shortest paths from the latter forced vertices go to $g_j$ and then to $\nw$ (resp. $\se$).
  Similarly in case $e_{j+1}$ has an endpoint in $V_i$, the shortest paths from the forced vertices attached to the neighbors of $c_{j+1}$ to $\{\nw,\nee\}$ (resp.~$\{\sw,\se\}$) go to $g_{j+1}$, then to $\text{nw}_i^{j+1}$ and $\nw$ (resp. then to $\text{se}_i^{j+1}$ and $\se$).
  
  Note that all the forced vertices attached to neighbors of $\pi_i^j$ and $\rho_i^j$ (resp.~$\pi_i^{j+1}$ and $\rho_i^{j+1}$) have a shortest path to $\{\nw,\nee\}$ ending in $\nw$ (resp.~to $\{\sw,\se\}$ ending in $\se$). 
  Finally due to the paths $P(\nu(\pi_i^j,\tri),\se)$ and $P(\nu(\rho_i^j,\tri),\se)$, all the forced vertices attached to neighbors of $\pi_i^j$ and $\rho_i^j$ have a shortest path to $\{\sw,\se\}$ ending in $\se$.
  And due to the paths $P(\nu(\pi_i^{j+1},\text{bl}_i^{j+1}),\nw)$ and $P(\nu(\rho_i^{j+1},\text{bl}_i^{j+1}),\nw)$, all the forced vertices attached to neighbors of $\pi_i^{j+1}$ and $\rho_i^{j+1}$ have a shortest path to $\{\nw,\nee\}$ ending in $\nw$.

  We now show that no critical pair $\{c_j,c'_j\}$ is resolved by a forced vertex.
  We set $Y' := \{\tle, \tri, \ble, \bri, \text{tl}_{i'}^j, \text{tr}_{i'}^j, \text{bl}_{i'}^j, \text{br}_{i'}^j, \nw, \se, \text{nw}_{i'}^j, \text{se}_{i'}^j, \pi_i^j, \rho_i^j, \pi_{i'}^j, \rho_{i'}^j\}$, with $e_j \in E(V_i,V_{i'})$, and $X'$ be the connected component of $G'' - Y'$ containing $g_j$.
  We showed in \cref{lem:onlyLocal}, and it remains true in $G''$, that no vertex of $Y' \setminus \{\pi_i^j, \rho_i^j, \pi_{i'}^j, \rho_{i'}^j\}$ resolves $\{c_j,c'_j\}$.
  We observe that $\pi_i^j$ and $\rho_i^j$ have shortest paths to $c_j$ going through $g_j$ (via a vertex of $\{\tle,\tri,\ble,\bri\}$).
  Similarly $\pi_{i'}^j$ and $\rho_{i'}^j$ have shortest paths to $c_j$ going through $g_j$.
  Therefore no vertex of $\{\pi_i^j, \rho_i^j, \pi_{i'}^j, \rho_{i'}^j\}$ resolves the pair $\{c_j,c'_j\}$.
  Hence by \cref{lem:surrounded}, no vertex outside $X'$ resolves $\{c_j,c'_j\}$.
  The only forced vertices in $X'$ are attached to neighbors of $g_j$, thus they do not resolve $\{c_j,c'_j\}$.
\end{proof}

\subsubsection{\smd-instance has a solution $\Rightarrow$ \sgmd-instance has a solution.}\label{sec:correctnessMDtoGMD}
Let $S$ be a resolving set for the \md-instance.
We show that $S' := S \cap \bigcup_{i \in [k], j \in [m]} V_i^j$ is a solution for \gmd.
The set $S \setminus S'$ is made of $14km+14m$ forced vertices, none of which is in some $V_i^j \cup P(v_{i,1}^j,p_i^j) \cup \{q_i^j\} \cup P(v_{i,t}^j,r_i^j) \cup \{s_i^j\}$.
Thus by \cref{lem:tooFar-useless}, $S \setminus S'$ does not resolve any pair $\{p_i^j, q_i^j\}$ or $\{r_i^j, s_i^j\}$.
Now $S'$ is a set of $k''-(14km+14m)=km$ vertices resolving all the $2km$ pairs $\{p_i^j, q_i^j\}$ and $\{r_i^j, s_i^j\}$.
Again by \cref{lem:tooFar-useless}, this is only possible if $|S' \cap V_i^j|=1$.
Thus $S'$ is a legal set of size $k' = km$.
Let us now check that $S'$ resolves every pair of $\mathcal P$ in the graph $G'$.

By \cref{lem:forcedAndCritical}, $S \setminus S'$ does not resolve any pair of $\mathcal P$ in the graph $G''$.
Thus $S'$ resolves all the pairs of $\mathcal P$ in $G''$.
Since the distances between $V_i^j$ and the critical pairs in the edge and propagation gadgets $V_i^j$ is attached to are the same in $G'$ and in $G''$, $S'$ also resolves every pair of $\mathcal P$ in $G'$.
Thus $S'$ is a solution for the \sgmd-instance.

\subsubsection{\sgmd-instance has a solution $\Rightarrow$ \smd-instance has a solution.}

For every $i \in [k]$, $j \in [m]$, let $$F_i^j := \bigcup_{u \in \{\nw,\se\} \cup N(\{\pi_i^j,\rho_i^j\}) \setminus \{p_i^j, q_i^j, r_i^j, s_i^j\}} \{f(u)\}\text{,~and}$$
$$F_j := \bigcup_{u \in N(g_j) \setminus \{c_j,c'_j\}} \{f(u)\}.$$  
Let $S$ be a solution for \gmd.
Thus $|S|=km$.
Let $F := \bigcup_{i \in [k], j \in [m]} F_i^j \cup \bigcup_{j \in [m]} F_j$.
We show that $S' := S \cup F$ is a solution of \md.
First we observe that $|S'|=km+14km+14m=k''$.
Since the distances between the sets $V_i^j$ and the critical pairs (of $\mathcal P$) are the same in $G'$ and in $G''$, the pairs of $\mathcal P$ are resolved by~$S$.
In what follows, we show that $F$ resolves all the other pairs.
For every $i \in [k]$, $j \in [m]$, we define the subset of vertices: $$\Pi_i^j := \bigcup_{u \in \{\tri,\tle,\bri,\ble,\bbr,\tc\}} P(\pi_i^j,u) \cup P(v_{i,1}^j,p_i^j) \cup \{q_i^j\},$$
$$R_i^j :=  \bigcup_{u \in \{\tri,\tle,\bri,\ble,\tbr,\bc\}} P(\rho_i^j,u) \cup P(v_{i,t}^j,r_i^j) \cup \{s_i^j\},~\text{and}$$
$$G_j := \bigcup_{u \in \{\tri,\tle,\bri,\ble,\text{tl}_{i'}^j,\text{tr}_{i'}^j,\text{bl}_{i'}^j,\text{br}_{i'}^j,\nw,\se,\text{nw}_{i'}^j,\text{se}_{i'}^j\}} P(g_j,u) \cup E_i^j \cup E_{i'}^j \cup \{c'_j\}.$$ 
Informally $\Pi_i^j$ ($R_i^j$, $G_j$, respectively) consists of the vertices on the paths incident to $\pi_i^j$ ($\rho_i^j$, $g_j$, respectively).
Our objective is the following result.

\begin{lemma}\label{thm:resolve}
Every vertex in $G''$ is distinguished by $S'$.
\end{lemma}

\begin{table}[h!]
  \centering
\begin{tabular}{|c|c|}
  \hline
  Symbol/term & Definition/action \\
  \hline
  $\{a_{i,\gamma}^j, \alpha_{i,\gamma}^j\}$ & critical pair of the propagation gadget $P_i^{j,j+1}$ \\
  $A_i^j$ & set of vertices $\bigcup_{\gamma \in [t]}\{a_{i,\gamma}^j, \alpha_{i,\gamma}^j\}$ \\
  $\bbr$ & bottom brown vertex, $\nu(v_{i,t}^j,r_i^j)$ \\
  $\bc$ & bottom cyan vertex (smallest index $\gamma$) \\
  $\ble$ & neighbor of $v_{i,t}^j$ in $P_i^{j-1,j}$ \\
  blue vertex & one of the four neighbors of $V_i^j$ in the propagation gadgets \\
  $\bri$ & neighbor of $v_{i,t}^j$ in $P_i^{j,j+1}$ \\
  brown vertex & vertices $\nu(v_{i,1}^j,p_i^j)$ and $\nu(v_{i,t}^j,r_i^j)$ \\
  $\{c_j,c'_j\}$ & critical pair of the edge gadget $\mathcal{G}(e_j)$ \\
  cyan vertex & neighbor of $V_i^j$ in the paths to $\mathcal{G}(e_j)$ \\
  $E_i^j$  & vertices in the paths from $V_i^j$ to $\mathcal{G}(e_j)$ \\
  $e_{i,\gamma}^j$ & alternative labeling of the cyan vertices, neighbor of $v_{i,\gamma}^j$ \\
  $F$ & set of all forced vertices, $\bigcup_{i \in [k], j \in [m]} F_i^j \cup \bigcup_{j \in [m]} F_j$ \\
  $F_i^j$ & set of forced vertices attached to neighbors of $\{\pi_i^j,\rho_i^j,\nw,\se\}$ \\
  $F_j$ & set of forced vertices attached to neighbors of $g_j$ \\
  $f(v)$ & forced vertex attached to a vertex $v$ \\
  $f'(v)$ & false twin of $f(v)$ \\
  $\mathcal{G}(e_j)$ & edge gadget on $\{g_j,c_j,c'_j\}$ between $V_i^j$ and $V_{i'}^j$, where $e_j \in E(V_i,V_{i'})$ \\
  $\text{mc}_i^j$ & middle cyan vertex (not top nor bottom) \\
  \nege  &  north-east gate of $P_i^{j,j+1}$ \\
  \nwge  &  north-west gate of $P_i^{j,j+1}$ \\
  \nele,  \swle  & resolve the critical pairs of $P_i^{j,{j+1}}$  \\
  \nwle, \sele & do not resolve the critical pairs of $P_i^{j,{j+1}}$ \\
  $\nu(u,v)$ & neighbor of $u$ in the path $P(u,v)$ \\
  $\mathcal{P}$ & list of critical pairs \\
  $\{p_i^j, q_i^j\}$ & pair only resolved by vertices in $V_i^j \cup P(v_{i,1}^j,p_i^j) \cup \{q_i^j\}$ \\
  $\pi_i^j$ & gate of $\{p_i^j, q_i^j\}$, linked by paths to most neighbors of $V_i^j$ \\
  $P_i^{j,j+1}$ & propagation gadget between $V_i^j$ and $V_i^{j+1}$ \\
  $P(u,v)$ & path added in the construction expressly between $u$ and $v$ \\
  $\{r_i^j, s_i^j\}$ & pair only resolved by vertices in $V_i^j \cup P(v_{i,t}^j,r_i^j) \cup \{s_i^j\}$ \\
  $\rho_i^j$ & gate of $\{r_i^j, s_i^j\}$, linked by paths to most neighbors of $V_i^j$ \\
  \sege  &  south-east gate of $P_i^{j,j+1}$ \\  
  \swge  &  south-west gate of $P_i^{j,j+1}$ \\                
  $t$ & size of each $V_i$ \\
  $\tbr$ & top brown vertex, $\nu(v_{i,1}^j,p_i^j)$ \\
  $\tc$ & top cyan vertex (largest index $\gamma$) \\
  $\tle$ & neighbor of $v_{i,1}^j$ in $P_i^{j-1,j}$\\
  $\tri$ & neighbor of $v_{i,1}^j$ in $P_i^{j,j+1}$\\
  $V_i$ & partite set of $G$ \\
  $V_i^j$ & ``copy of $V_i$'', stringed by a path, in $G'$ and $G''$ \\
  $v_{i,\gamma}^j$ & vertex of $V_i^j$ representing $v_{i,\gamma} \in V(G)$ \\
  $W_j$ & endpoints in $V_i^j \cup V_{i'}^j$ of paths from $V_i^j \cup V_{i'}^j$ to $\mathcal{G}(e_j)$ \\
  $\mathcal X$ & set containing all the sets $V_i^j$ for $i \in [k]$ and $j \in [m]$ \\
  $X_j$ & neighbors of $W_j$ on the paths to $\mathcal{G}(e_j)$ (cyan vertices) \\
  \hline
\end{tabular}
\caption{Glossary of the construction.}
\label{tbl:glossary}
\end{table}

We start with the forced vertices and their false twin.
We denote by \emph{$f'(v)$} the false twin of the forced vertex $f(v)$.
\begin{lemma}\label{lem:twinsDistinguished}
All the vertices $f(v)$ and $f'(v)$ are distinguished by $F$.
\end{lemma}
\begin{proof}
  Any vertex $f(v)$ is distinguished by being the only vertex at distance $0$ of itself $f(v) \in F$.
  Since $f(v)$ has only two neighbors $f'(v)$ and $v$, it also resolves every pair $\{f'(v),w\}$ where $w$ is not $v$.
  The pair $\{f'(v),v\}$ is resolved by any vertex $f \in F \setminus \{f(v)\}$.
  Indeed $\dist(f,f'(v))=\dist(f,v)+1$.
  Thus $f'(v)$ is distinguished.
\end{proof}

In general, to show that all the vertices in a set $X$ are distinguished, we proceed in two steps.
First we show that every internal pair of $X$ is resolved.
Then, we prove that every pair of $X \times \overline{X}$ is also resolved.
Let us recall that $\overline X$ is the complement of $x$, here $V(G'') \setminus X$.
For instance, the two following lemmas show that every vertex of $\Pi_i^j$ is distinguished by $S'$.

\begin{lemma}\label{lem:internal pair resolve}
  Every pair of distinct vertices $x, y \in \Pi_i^j$ is resolved by $S'$.
\end{lemma}
\begin{proof}
  Let $U_i^j$ be the set $\{\tle,\tri,\ble,\bri,\tc,\bbr\}$.
  We first consider two vertices $x \neq y \in P(\pi_i^j,u)$, for some $u \in U_i^j$.
  As $\dist_{G''}(\pi_i^j,u)$ is equal to the length of $P(\pi_i^j,u)$, it holds that $\dist_{G''}(\pi_i^j,x) = \dist_{P(\pi_i^j,u)}(\pi_i^j,x) \neq \dist_{P(\pi_i^j,u)}(\pi_i^j,y) = \dist_{G''}(\pi_i^j,y)$.
  Without loss of generality, we assume that $\dist(\pi_i^j,x) < \dist(\pi_i^j,y)$. 
  If $x \neq \pi_i^j$, then $x$ and $y$ have distinct distances to $\nu(\pi_i^j,u)$.
  Hence $\dist(f(\nu(\pi_i^j,u)),x) \neq \dist(f(\nu(\pi_i^j,u)),y)$ and $S'$ resolves $\{x,y\}$.
  Now if $x = \pi_i^j$, then $f(\nu(\pi_i^j,u'))$ resolves $\{x,y\}$ for any $u' \in U_i^j \setminus \{u\}$.
 
 Secondly we consider $x \in P(\pi_i^j,u)$ and $y \in P(\pi_i^j,u')$, for some $u \neq u' \in U_i^j$.
 If $\dist(\pi_i^j,x) \neq 2+\dist(\pi_i^j,y)$, then $f(\nu(\pi_i^j,x))$ resolves $\{x,y\}$.
 Indeed $\dist(f(\nu(\pi_i^j,x)),x) = \dist(\pi_i^j,x) \neq 2 + \dist(\pi_i^j,y) = \dist(f(\nu(\pi_i^j,x)),y)$.
 Else if $\dist(\pi_i^j,x) = 2+\dist(\pi_i^j,y)$, then $f(\nu(\pi_i^j,y))$ resolves $\{x,y\}$ (since $\dist(\pi_i^j,y) \neq 2+\dist(\pi_i^j,x)$).

 Two distinct vertices on $P(v_{i,1}^j,p_i^j)$ are resolved by, say, $f(\nu(\pi_i^j,\bri)) \in F$.
 A vertex of $P(v_{i,1}^j,p_i^j)$ and a vertex of $P(\pi_i^j,u)$, for some $u  \in U_i^j$, are resolved by either $f(\nu(\pi_i^j,u))$ or $f(\nu(\pi_i^j,u'))$ for a $u' \in U_i^j \setminus \{u\}$.
 Finally $q_i^j$ and a vertex in $P(v_{i,1}^j,p_i^j) \setminus \{p_i^j\}$ are resolved by, say, $f(\nu(\pi_i^j,\bri))$, whereas $q_i^j$ and a vertex in $P(p_i^j,u)$ is resolved by either $f(\nu(\pi_i^j,u))$ or $f(\nu(\pi_i^j,u'))$ for a $u' \in U_i^j \setminus \{u\}$. 
 Therefore every pair of distinct vertices in $\Pi_i^j$ is resolved by~$F$, except $\{p_i^j,q_i^j\}$ which is resolved by~$S$.
\end{proof}

\begin{lemma}\label{lem:internal-external pair resolve}
  Every pair $\{x,y\} \in \Pi_i^j \times \overline{\Pi_i^j}$ is resolved by $F$.
\end{lemma}
\begin{proof}
  Again let $U_i^j$ be the set $\{\tle,\tri,\ble,\bri,\tc,\bbr\}$.
  We first assume $x$ is in $P(\pi_i^j,u)$ for some $u \in U_i^j \setminus \{\tri,\ble\}$.
  Let $y$ be a vertex of $\overline{\Pi_i^j}$ such that $\dist(f(\nu(\pi_i^j,u)),x) = \dist(f(\nu(\pi_i^j,u)),y)$, otherwise $f(\nu(\pi_i^j,u))$ already resolves $\{x,y\}$.
  Every shortest path from $f(\nu(\pi_i^j,u))$ to $y$ go through $\pi_i^j$.
  One can observe that there is a $u' \in U_i^j \setminus \{u\}$ such that $f(\nu(\pi_i^j,u'))$ has a shortest path also going through $\pi_i^j$.
  Hence $f(\nu(\pi_i^j,u'))$ has the same distance to $y$ (as $f(\nu(\pi_i^j,u))$) but a larger distance to $x$.
  Hence $f(\nu(\pi_i^j,u'))$ resolves $\{x,y\}$.

  We now consider an $x \in P(\pi_i^j,u)$ for some $u \in \{\tri,\ble\}$.
  Again let $y$ be a vertex of $\overline{\Pi_i^j}$ such that $\dist(f(\nu(\pi_i^j,u)),x) = \dist(f(\nu(\pi_i^j,u)),y)$.
  If all the shortest paths of $f(\nu(\pi_i^j,u))$ to $y$ goes through $\pi_i^j$, we conclude as in the previous paragraph. 
  So they go through $P(\nu(\pi_i^j,u),\se)$ (if $u=\tri$) or $P(\nu(\pi_i^j,u),\text{nw}_i^{j-1})$ (if $u=\ble$).
  Since $\dist(f(\nu(\pi_i^j,u)),x) \leqslant 2t-1$, it also holds that $\dist(f(\nu(\pi_i^j,u)),y) \leqslant 2t-1$.
  The path $P(\nu(\pi_i^j,\tri),\se)$ has length $t$ and the path $P(\nu(\pi_i^j,\ble),\text{nw}_i^{j-1})$ has length $2t-1$.
  Therefore one of $f(\se)$, $f(\text{se}_i^{j-1})$, $f(\nw)$, $f(\text{nw}_i^{j-1})$ resolves $\{x,y\}$.

  We now assume $x$ is in $P(v_{i,1}^j,p_i^j) \cup \{q_i^j\}$ and $y \in \overline{\Pi_i^j}$.
  Then $f(\nu(\pi_i^j,\bri))$ resolves $\{x,y\}$ if $y$ is not in the path $P(\nu(\pi_i^j,\tri),\se)$ or $P(\nu(\pi_i^j,u),\text{nw}_i^{j-1})$.
  Otherwise at least one of $f(\nu(\pi_i^j,\bri))$, $f(\nu(\pi_i^j,\tri))$, $f(\nu(\pi_i^j,\ble))$ resolves $\{x,y\}$.
  In conclusion, every pair of vertices $\{x,y\} \in \Pi_i^j \times \overline {\Pi_i^j}$ is resolved by $F$. 
 \end{proof}
 
 \cref{lem:internal pair resolve,lem:internal-external pair resolve} prove that every vertex in $\Pi_i^j$ is distinguished by $S'$.
 Using the same arguments, we get symmetrically that every vertex of $R_i^j$ is distinguished by $S'$.

 \begin{lemma}\label{lem:finishingTouchesDistinguished}
 All the vertices in the paths $P(\nu(\pi_i^j,\tri),\se)$, $P(\nu(\rho_i^j,\tri),\se)$, $P(\nu(\pi_i^j,\ble),$ $\text{nw}_i^{j-1})$, $P(\nu(\rho_i^j,\ble),\text{nw}_i^{j-1})$ are distinguished by $F$.
 \end{lemma}
 \begin{proof}
   Any vertex $x \in P(\nu(\pi_i^j,\tri),\se)$ is uniquely determined by its distances to $f(\se)$, $f(\text{se}_i^{j-1})$, and $\nu(\pi_i^j,\tri)$.
   Any vertex $x \in P(\nu(\pi_i^j,\ble),\text{nw}_i^{j-1})$ is uniquely determined by its distances to $f(\nw)$, $f(\text{nw}_i^{j-1})$, and $\nu(\pi_i^j,\ble)$.
   The two other cases are symmetric.
 \end{proof}

 So far we showed that the vertices added in the forced set and forced vertex gadgets are all distinguished.
 We now focus on the vertices in propagation gadgets.
 Let $\Delta_i := A_i^j \cup \{\nw,\nee,\sw,\se\}$.
 
\begin{lemma}\label{lem:internal pair resolve.}
 Every pair of distinct vertices $x, y \in \Delta_i^j$ is resolved by $S'$.
\end{lemma}  
\begin{proof}
  Since the distances between vertices of $V_i^j$ and vertices of $\Delta_i^j$ are the same between $G'$ and $G''$, $S$ resolves all the critical pairs $\{a_{i,\gamma}^j,\alpha_{i,\gamma}^j\}$.
  Thus we turn our attention to the pairs which are not critical pairs.
  Since $\dist(\nw,a_{i,\gamma}^j)=\gamma$ and $\dist(\nw,\alpha_{i,\gamma}^j)=\gamma$, every pair $\{a_{i,\gamma}^j,a_{i,\gamma'}^j\}$, $\{a_{i,\gamma}^j,\alpha_{i,\gamma'}^j\}$, or $\{\alpha_{i,\gamma}^j,\alpha_{i,\gamma'}^j\}$, with $\gamma \neq \gamma'$ is resolved by $f(\nw)$.
  
  Gate $\nw$ (resp.~$\se$) and any other vertex in $\Delta_i^j$ is resolved by $f(\nw)$ (resp.~$f(\se)$).
  Gate $\nee$ (resp.~$\sw$) is resolved from any vertex of $\Delta_i^j \setminus \{a_{i,1}^j,\alpha_{i,1}^j\}$ (resp.~$\Delta_i^j \setminus \{a_{i,t}^j,\alpha_{i,t}^j\}$) by $f(\nw)$ (resp.~$f(\se)$).
  Finally, $\nee$ (resp.~$\sw$) and a vertex of $\{a_{i,1}^j,\alpha_{i,1}^j\}$ (resp.~$\{a_{i,t}^j,\alpha_{i,t}^j\}$) is resolved by $f(\se)$ (resp.~$f(\nw)$).
\end{proof}

Now when we check that a pair made of a vertex in $\Delta_i^j$ and a vertex outside $\Delta_i^j$ is resolved, we can further assume that the second vertex is not in some $\Pi_i^j \cup R_i^j$ since we already showed that these vertices were distinguished.

\begin{lemma}\label{lem:internal -external pair resolve.}
Every pair $\{x,y\} \in \Delta_i^j \times \overline{\Delta_i^j}$ is resolved by $S'$.
\end{lemma}

\begin{proof}
  We may assume that $y$ is not a vertex that was previously shown distinguished.
  Thus $y$ is not in some $\Pi_i^j \cup R_i^j$ nor in a path of \cref{lem:finishingTouchesDistinguished}.
  Then we claim that the pair $\{x,y\}$ is resolved by at least one of $f(\se)$, $f(\text{se}_i^{j-1})$, $f(\text{se}_i^{j+1})$, $f(\nw)$.
  Indeed assume that $f(\se)$ does not resolve $\{x,y\}$, and consider a shortest path from $f(\se)$ to $y$.
  Either this shortest path goes through $\text{se}_i^{j-1}$ (resp.~$\text{se}_i^{j+1}$), and in that case $f(\text{se}_i^{j-1})$ (resp.~$f(\text{se}_i^{j+1})$) resolves $\{x,y\}$.
  Either it takes the path to $g_j$ (if $e_j$ has an endpoint in $V_i$) or to $\text{tl}_i^{j+1}$, and then $f(\nw)$ resolves $\{x,y\}$.
  Or it takes a path to $V_i^j$, and then $f(\text{se}_i^{j-1})$ resolves $\{x,y\}$.
\end{proof}

\cref{lem:internal pair resolve.,lem:internal -external pair resolve.} show that that every vertex in $\Delta_i^j$ is distinguished by $S'$.
The common neighbor of $\text{se}_i^{j-1}$ and $\tle$ is distinguished by $\{f(\text{se}_i^{j-1}),f(\nu(\pi_i^j,\tle))\}$.
We are now left with showing that the vertices in the edge gadgets, in the sets $V_i^j$, and in the paths incident to the edge gadgets, are distinguished.

\begin{lemma} \label{lem:internal}
Every pair of distinct vertices $x, y \in G_j$ is resolved by $S'$.
\end{lemma}
\begin{proof}
  Let $v_{i,\gamma}$ and $v_{i',\gamma'}$ be the two endpoints of $e_j$, and $U_i^j := \{\tle,\tri,\ble,\bri,\text{tl}_{i'}^j,\text{tr}_{i'}^j,\text{bl}_{i'}^j,$ $\text{bl}_{i'}^j,\nw,\se,\text{nw}_{i'}^j,\text{se}_{i'}^j,v_{i,\gamma}^j,v_{i',\gamma'}^j\}$.
  Every pair in $\bigcup_{u \in U_i^j} P(g_j,u)$ is resolved.
  Indeed, similarly to \cref{lem:internal pair resolve}, two distinct vertices $x, y$ on a path $P(g_j,u)$ ($u \in U_i^j$) are resolved by $f(\nu(g_j,u))$, while two vertices on distinct paths $P(g_j,u)$ and $P(g_j,u')$ ($u \neq u' \in U_i^j$) are resolved by at least one of $f(\nu(g_j,u))$ and $f(\nu(g_j,u'))$.

  We now show that any pair in $\Gamma_i^j := E_i^j \cup E_{i'}^j \setminus \{P(g_j,v_{i,\gamma}^j),P(g_j,v_{i,\gamma}^j)\}$ is resolved.
  Two distinct vertices $x, y \in \Gamma_i^j$ are resolved by, say, $f(\nu(g_j,\se))$ if they are on the same path, or more generally if they have different distances to $c_j$.
  Thus let us assume that $x$ and $y$ are at the same distance from $c_j$.
  If $x \in E_i^j$ and $y \in E_{i'}^j$ (or vice versa) then the pair $\{x,y\}$ is resolved by the vertex in $S \cap V_i^j$ or the vertex in $S \cap V_{i'}^j$.
  If $x \neq y \in E_i^j$ (resp. $\in E_{i'}^j$), then $\{x,y\}$ is resolved by $f(\nu(g_j,\nw))$ (resp.~$f(\nu(g_j,\text{nw}_{i'}^j))$).
  This is the reason why we added an edge between $\nu(g_j,\nw)$ and $\nu(c_j,\bc)$ (recall \cref{subsec:finalTouches}).

  We now consider pairs $\{x,y\}$ of $\bigcup_{u \in U_i^j} P(g_j,u) \times \Gamma_i^j$.
  Any of these pairs are resolved by at least one of $f(\nu(g_j,u))$, $f(\nu(g_j,u'))$, $f(\nu(g_j,\nw))$, $f(\nu(g_j,\text{nw}_{i'}^j))$, where $x$ is on the path $P(c_j,u)$ and $u'$ is any vertex in $U_i^j \setminus \{u,\nw,\text{nw}_{i'}^j\}$.
  Finally $c'_j$ is distinguished from all the other vertices in $G''$ but $c_j$ by the forced vertices attached to the neighbors of $g_j$. 
 
  Thus every pair $\{x,y\}$ in $G_j$ is resolved by $F$, except $\{c_j,c'_j\}$ which is resolved by $S$. 
\end{proof}

\begin{lemma}\label{lem:internal-external}
Every pair $\{x,y\} \in G_j \times \overline{G_j}$ is resolved by $F$.
\end{lemma}

\begin{proof}
  Consider an arbitrary pair $\{x,y\} \in G_j \times \overline{G_j}$.
  We can assume that $x$ is not $c'_j$, and that $y$ is in one different $G_{j'}$ or in one $V_{i''}^{j''}$ (since we already showed that the other vertices are distinguished).
  Again let $v_{i,\gamma}$ and $v_{i',\gamma'}$ be the two endpoints of $e_j$, and $U_i^j := \{\tle,\tri,\ble,\bri,\text{tl}_{i'}^j,\text{tr}_{i'}^j,\text{bl}_{i'}^j,$ $\text{bl}_{i'}^j,\nw,\se,\text{nw}_{i'}^j,\text{se}_{i'}^j,v_{i,\gamma}^j,v_{i',\gamma'}^j\}$.
  If $x$ is on a path $P(g_j,u)$, then at least one of $f(\nu(g_j,u))$ and $f(\nu(g_j,u'))$, with $u'$ being any vertex in $U_i^j \setminus \{u\}$, resolves $\{x,y\}$.
  If instead $x$ is on a path $P(c_j,u)$ with $u \in \{v_{i,\gamma-1}^j,v_{i,\gamma+1}^j,v_{i',\gamma'-1}^j,v_{i',\gamma'+1}^j\}$, then at least one of $f(\nu(g_j,\nw))$, $f(\nu(g_j,\text{nw}_{i'}^j))$, $f(\nu(g_j,u'))$, with $u'$ being any vertex in $U_i^j$, resolves $\{x,y\}$.
\end{proof}

 \cref{lem:internal,lem:internal-external} show that every vertex in $G_j$ is distinguished by $S'$.
 We finally show that the vertices in $V_i^j$ are distinguished.
 A pair of distinct vertices $x, y \in V_i^j$ is resolved by $f(\nw)$.
 We thus consider a pair $\{x, y\} \in V_i^j \times \overline{V_i^j}$.
 We can further assume that $y$ is in some $V_{i'}^{j'}$, since all the other vertices have already been shown distinguished.
 Then $\{x,y\}$ is resolved by at least one of $f(\nw)$, $f(\text{nw}_{i'}^{j'})$, the vertex in $S \cap V_i^j$, and the vertex in $S \cap V_{i'}^{j'}$.
 This finishes the proof of \cref{thm:resolve}.
 Thus $S'$ is a solution of the \md-instance.
 
 The reduction is correct and it takes polynomial-time in $|V(G)|$ to compute $G''$.
 The maximum degree of $G''$ is $16$.
 It is the degree of the vertices $g_j$ ($\nw$ and $\se$ have degree at most $11$, $\pi_i^j$ and $\rho_i^j$ have degree $8$, and the other vertices have degree at most $5$).
 The last element to establish \cref{thm:main} is to show that $\pw(G'')$ is in $O(k)$.
 Then solving \md on constant-degree graphs in time $f(\pw)n^{o(\pw)}$ could be used to solve \kmIS in time $f(k)n^{o(k)}$, disproving the ETH. 
 
 \subsection{$G''$ has pathwidth $O(k)$}

 We use the pathwidth characterization of Kirousis and Papadimitriou \cite{Kirousis85} mentioned in the preliminaries, and give a strategy with $O(k)$ searchers cleaning all the edges of $G''$.
 A basic and useful fact is that the searching number of a path is two. 

\begin{lemma}\label{lem:node searching of a path}
Two searchers are enough to clean a path $u_1u_2 \ldots u_n$.
\end{lemma}
\begin{proof}  
We place two searchers at $u_1$ and $u_2$.
This cleans the edge $u_1u_2$.
Then we move the searcher in $u_1$ to $u_3$.
This cleans $u_2u_3$ (while $u_1u_2$ remains clean).
Then we move the searcher in $u_2$ to $u_4$, and so on.
\end{proof}

\begin{lemma}\label{lem:pathwidth}
$\pw(G'') \leqslant 90k+83$.
 \end{lemma}
\begin{proof}
  For every $j \in [m]$, let $S_j := N[g_j]~\cup~X_j~\cup~\bigcup_{i \in [k]} N[\{v_{i,1}^j,v_{i,t}^j,\pi_i^j,\rho_i^j\}] \cup \{\nw,\nee,\sw,$ $\se\}$.
  We notice that $|S_j| \leqslant 17+6+30k+4=30k+27$.
  Another important observation is that $S_1 \cup S_j$ disconnects the first $j$ columns of $G''$ from the rest of $G''$.
  Finally the connected components $G''-(S_j \cup S_{j+1})$ that are not the main component (i.e., containing more than half of the graph if $m \geqslant 4$) are all paths.

  We now suggest the following cleaning strategy with at most $90k+83$ searchers.
  We place one searcher at each vertex of $S_1 \cup S_2 \cup S_3$.
  This requires $90k+81$ searchers.
  By \cref{lem:node searching of a path}, with two additional searchers we clean all the connected components of $G'' - (S_1 \cup S_2 \cup S_3)$ that are paths.
  We then move all the searchers from $S_2$ to $S_4$, and clean all the connected components of $G'' - (S_1 \cup S_3 \cup S_4)$ that are paths.
  Since $S_1 \cup S_3$ is a separator, the edges that were cleaned during the first phase are not recontaminated when we move from $S_2$ to $S_4$.
  We then move the searchers of $S_3$ to $S_5$, and so on.
  Eventually the searchers reach $S_1 \cup S_{m-1} \cup S_m$, and the last contaminated edges are cleaned.
\end{proof}

\section{Perspectives}\label{sec:perspectives}
The main remaining open question is whether or not \md is polytime solvable on graphs with constant treewidth.
In the parameterized complexity language, now we know that \smd/$\tw$ is W[1]-hard, is it in XP or paraNP-hard?  
We believe that the tools and ideas developed in this paper could help answering this question negatively.
The FPT algorithm of Belmonte et al. \cite{Belmonte17} also implies that \md is FPT with respect to $\tl+k$ were $k$ is the size of the resolving set, due to the bound $\Delta \leqslant 2^k+k-1$ \cite{Khuller96}.
What about the parameterized complexity of \md with respect to $\tw+k$?
We conjecture that this problem is W[1]-hard as well, and once again, treewidth will contrast with tree-length.

It appears that bounded connected treewidth or tree-length is significantly more helpful than the mere bounded treewidth when it comes to solving \smd.
We wish to ask for the parameterized complexity of \md with respect to $\ctw$ only (on graphs with arbitrarily large degree).
Finally, it would be interesting to determine if planarity can sometimes help to compute a metric basis.
Therefore we also ask all the above questions in planar graphs.


\begin{thebibliography}{10}

\bibitem{Beerliova06}
Zuzana Beerliova, Felix Eberhard, Thomas Erlebach, Alexander Hall, Michael
  Hoffmann, Mat{\'{u}}s Mihal{\'{a}}k, and L.~Shankar Ram.
\newblock Network discovery and verification.
\newblock {\em {IEEE} Journal on Selected Areas in Communications},
  24(12):2168--2181, 2006.
\newblock URL: \url{https://doi.org/10.1109/JSAC.2006.884015}, \href
  {http://dx.doi.org/10.1109/JSAC.2006.884015}
  {\path{doi:10.1109/JSAC.2006.884015}}.

\bibitem{Belmonte17}
R{\'{e}}my Belmonte, Fedor~V. Fomin, Petr~A. Golovach, and M.~S. Ramanujan.
\newblock Metric dimension of bounded tree-length graphs.
\newblock {\em {SIAM} J. Discrete Math.}, 31(2):1217--1243, 2017.
\newblock URL: \url{https://doi.org/10.1137/16M1057383}, \href
  {http://dx.doi.org/10.1137/16M1057383} {\path{doi:10.1137/16M1057383}}.

\bibitem{Chartrand00}
Gary Chartrand, Linda Eroh, Mark~A. Johnson, and Ortrud Oellermann.
\newblock Resolvability in graphs and the metric dimension of a graph.
\newblock {\em Discrete Applied Mathematics}, 105(1-3):99--113, 2000.
\newblock URL: \url{https://doi.org/10.1016/S0166-218X(00)00198-0}, \href
  {http://dx.doi.org/10.1016/S0166-218X(00)00198-0}
  {\path{doi:10.1016/S0166-218X(00)00198-0}}.

\bibitem{Chvatal83}
Vasek Chv{\'{a}}tal.
\newblock Mastermind.
\newblock {\em Combinatorica}, 3(3):325--329, 1983.
\newblock URL: \url{https://doi.org/10.1007/BF02579188}, \href
  {http://dx.doi.org/10.1007/BF02579188} {\path{doi:10.1007/BF02579188}}.

\bibitem{Cygan15}
Marek Cygan, Fedor~V. Fomin, Lukasz Kowalik, Daniel Lokshtanov, D{\'{a}}niel
  Marx, Marcin Pilipczuk, Michal Pilipczuk, and Saket Saurabh.
\newblock {\em Parameterized Algorithms}.
\newblock Springer, 2015.
\newblock URL: \url{https://doi.org/10.1007/978-3-319-21275-3}, \href
  {http://dx.doi.org/10.1007/978-3-319-21275-3}
  {\path{doi:10.1007/978-3-319-21275-3}}.

\bibitem{Diaz17}
Josep D{\'{\i}}az, Olli Pottonen, Maria~J. Serna, and Erik~Jan van Leeuwen.
\newblock Complexity of metric dimension on planar graphs.
\newblock {\em J. Comput. Syst. Sci.}, 83(1):132--158, 2017.
\newblock URL: \url{https://doi.org/10.1016/j.jcss.2016.06.006}, \href
  {http://dx.doi.org/10.1016/j.jcss.2016.06.006}
  {\path{doi:10.1016/j.jcss.2016.06.006}}.

\bibitem{DowneyF13}
Rodney~G. Downey and Michael~R. Fellows.
\newblock {\em Fundamentals of Parameterized Complexity}.
\newblock Texts in Computer Science. Springer, 2013.
\newblock URL: \url{https://doi.org/10.1007/978-1-4471-5559-1}, \href
  {http://dx.doi.org/10.1007/978-1-4471-5559-1}
  {\path{doi:10.1007/978-1-4471-5559-1}}.

\bibitem{Eppstein15}
David Eppstein.
\newblock Metric dimension parameterized by max leaf number.
\newblock {\em J. Graph Algorithms Appl.}, 19(1):313--323, 2015.
\newblock URL: \url{https://doi.org/10.7155/jgaa.00360}, \href
  {http://dx.doi.org/10.7155/jgaa.00360} {\path{doi:10.7155/jgaa.00360}}.

\bibitem{Epstein15}
Leah Epstein, Asaf Levin, and Gerhard~J. Woeginger.
\newblock The (weighted) metric dimension of graphs: Hard and easy cases.
\newblock {\em Algorithmica}, 72(4):1130--1171, 2015.
\newblock URL: \url{https://doi.org/10.1007/s00453-014-9896-2}, \href
  {http://dx.doi.org/10.1007/s00453-014-9896-2}
  {\path{doi:10.1007/s00453-014-9896-2}}.

\bibitem{Fernau15}
Henning Fernau, Pinar Heggernes, Pim van~'t Hof, Daniel Meister, and Reza Saei.
\newblock Computing the metric dimension for chain graphs.
\newblock {\em Inf. Process. Lett.}, 115(9):671--676, 2015.
\newblock URL: \url{https://doi.org/10.1016/j.ipl.2015.04.006}, \href
  {http://dx.doi.org/10.1016/j.ipl.2015.04.006}
  {\path{doi:10.1016/j.ipl.2015.04.006}}.

\bibitem{Foucaud17}
Florent Foucaud, George~B. Mertzios, Reza Naserasr, Aline Parreau, and Petru
  Valicov.
\newblock Identification, location-domination and metric dimension on interval
  and permutation graphs. {II.} algorithms and complexity.
\newblock {\em Algorithmica}, 78(3):914--944, 2017.
\newblock URL: \url{https://doi.org/10.1007/s00453-016-0184-1}, \href
  {http://dx.doi.org/10.1007/s00453-016-0184-1}
  {\path{doi:10.1007/s00453-016-0184-1}}.

\bibitem{GareyJ79}
Michael~R. Garey and David~S. Johnson.
\newblock {\em Computers and Intractability: {A} Guide to the Theory of
  NP-Completeness}.
\newblock W. H. Freeman, 1979.

\bibitem{Harary76}
Frank Harary and Robert~A Melter.
\newblock On the metric dimension of a graph.
\newblock {\em Ars Combin}, 2(191-195):1, 1976.

\bibitem{Hartung13}
Sepp Hartung and Andr{\'{e}} Nichterlein.
\newblock On the parameterized and approximation hardness of metric dimension.
\newblock In {\em Proceedings of the 28th Conference on Computational
  Complexity, {CCC} 2013, K.lo Alto, California, USA, 5-7 June, 2013}, pages
  266--276, 2013.
\newblock URL: \url{https://doi.org/10.1109/CCC.2013.36}, \href
  {http://dx.doi.org/10.1109/CCC.2013.36} {\path{doi:10.1109/CCC.2013.36}}.

\bibitem{Hoffmann16}
Stefan Hoffmann, Alina Elterman, and Egon Wanke.
\newblock A linear time algorithm for metric dimension of cactus block graphs.
\newblock {\em Theor. Comput. Sci.}, 630:43--62, 2016.
\newblock URL: \url{https://doi.org/10.1016/j.tcs.2016.03.024}, \href
  {http://dx.doi.org/10.1016/j.tcs.2016.03.024}
  {\path{doi:10.1016/j.tcs.2016.03.024}}.

\bibitem{Hoffmann12}
Stefan Hoffmann and Egon Wanke.
\newblock Metric dimension for gabriel unit disk graphs is {NP}-complete.
\newblock In {\em Algorithms for Sensor Systems, 8th International Symposium on
  Algorithms for Sensor Systems, Wireless Ad Hoc Networks and Autonomous Mobile
  Entities, {ALGOSENSORS} 2012, Ljubljana, Slovenia, September 13-14, 2012.
  Revised Selected Papers}, pages 90--92, 2012.
\newblock URL: \url{https://doi.org/10.1007/978-3-642-36092-3\_10}, \href
  {http://dx.doi.org/10.1007/978-3-642-36092-3\_10}
  {\path{doi:10.1007/978-3-642-36092-3\_10}}.

\bibitem{ImpagliazzoETH}
Russell Impagliazzo, Ramamohan Paturi, and Francis Zane.
\newblock Which problems have strongly exponential complexity?
\newblock {\em Journal of Computer and System Sciences}, 63(4):512--530,
  December 2001.

\bibitem{Khuller96}
Samir Khuller, Balaji Raghavachari, and Azriel Rosenfeld.
\newblock Landmarks in graphs.
\newblock {\em Discrete Applied Mathematics}, 70(3):217--229, 1996.
\newblock URL: \url{https://doi.org/10.1016/0166-218X(95)00106-2}, \href
  {http://dx.doi.org/10.1016/0166-218X(95)00106-2}
  {\path{doi:10.1016/0166-218X(95)00106-2}}.

\bibitem{Kirousis85}
Lefteris~M. Kirousis and Christos~H. Papadimitriou.
\newblock Interval graphs and searching.
\newblock {\em Discrete Mathematics}, 55(2):181--184, 1985.
\newblock URL: \url{https://doi.org/10.1016/0012-365X(85)90046-9}, \href
  {http://dx.doi.org/10.1016/0012-365X(85)90046-9}
  {\path{doi:10.1016/0012-365X(85)90046-9}}.

\bibitem{surveyETH}
Daniel Lokshtanov, D{\'{a}}niel Marx, and Saket Saurabh.
\newblock Lower bounds based on the exponential time hypothesis.
\newblock {\em Bulletin of the {EATCS}}, 105:41--72, 2011.
\newblock URL: \url{http://eatcs.org/beatcs/index.php/beatcs/article/view/92}.

\bibitem{Pietrzak03}
Krzysztof Pietrzak.
\newblock On the parameterized complexity of the fixed alphabet shortest common
  supersequence and longest common subsequence problems.
\newblock {\em J. Comput. Syst. Sci.}, 67(4):757--771, 2003.
\newblock URL: \url{https://doi.org/10.1016/S0022-0000(03)00078-3}, \href
  {http://dx.doi.org/10.1016/S0022-0000(03)00078-3}
  {\path{doi:10.1016/S0022-0000(03)00078-3}}.

\bibitem{Slater75}
Peter~J Slater.
\newblock Leaves of trees.
\newblock {\em Congr. Numer}, 14(549-559):37, 1975.

\end{thebibliography}

\end{document}